\documentclass[11pt,letterpaper]{article}
\usepackage[left=1in, right=1in, top=1in, bottom=1in]{geometry}
\usepackage{local}
\usepackage{algorithm}
\usepackage{algpseudocode}

\providecommand{\Description}[1]{}

\title{Facility Location Mechanism Design:\\ Breaking The Deterministic Barrier}
\author{
Zohar Barak\thanks{Tel Aviv University, \url{zoharbarak@mail.tau.ac.il}}}
\date{}

\begin{document}

\maketitle
\pagenumbering{gobble}

\begin{abstract}
We study the canonical problem of facility location mechanism design where as input we have $n$ locations in Euclidean space reported by strategic agents, and the output is a single facility location in space. The cost function of each agent is the distance from the returned facility, and the objective is to minimize the social cost function which is the sum of agent costs in a strategyproof way.

Our contributions are the following:
\begin{enumerate}
    \item \textbf{Breaking the deterministic barrier}. For $\R^2$, we give a random anonymous strategyproof mechanism (RR-CWM) that achieves an expected approximation ratio of $\frac{4}{\pi} \approx 1.27$, which strictly improves upon the best deterministic strategyproof mechanism that has an approximation ratio of $\sqrt{2} \approx 1.41$. Thus, this result closes the open problem of separating between deterministic and random mechanisms for utilitarian facility location mechanism design in $\R^2$.
    For $\R^d$, we show that the expected approximation ratio of our mechanism is in $[1.41 - O(1/\sqrt{d}), 1.547]$.

    \item \textbf{Improved learning augmented mechanisms through randomization.} We show that our ideas may be utilized to also achieve better performance in the learning augmented setting in $\R^2$, where in addition to the input the mechanism also receives predictions. For the output prediction model of Agrawal et al. 2022 we show an improved expected consistency-robustness trade-off. Our results also imply improved performance for the input MAC predictions model of Barak et al. 2024.
    \item \textbf{The limitations of Random Dictators}. We show a lower bound on the performance of the very common class of GRD (Generalized Random Dictator) mechanisms, where only locations reported by the agents may be returned. Our result states that any GRD mechanism has a larger expected approximation ratio than our RR-CWM mechanism, as our lower bound for $\R^2$ is $\frac{4}{\pi}$ (matching the upper bound of RR-CWM, which is not a GRD mechanism). For $\R^d$, we show a lower bound of $\sqrt{2} - O(1/d)$.
\end{enumerate}

\end{abstract}
\newpage
\tableofcontents

\newpage
\pagenumbering{arabic}
\setcounter{page}{1}
\newpage
\section{Intro}

The design of strategyproof mechanisms for facility location is a fundamental problem at the interface of algorithms, game theory, and social choice. A set of agents report locations in Euclidean space, a mechanism outputs a facility location, each agent’s cost is the distance to that facility, and the utilitarian objective is to minimize total (social) cost. The main tension is between \emph{efficiency} and \emph{incentives}: we want a mechanism that is strategyproof while approximating the optimal social cost well.

As mentioned by \citet{gravin2025approximation}, this problem was studied in earlier work in the field of social choice~\cite{moulin1980strategy,border1983straightforward,kim1984nonmanipulability,peters1993range,barbera1993generalized,ching1997strategy,peremans1997strategy,barbera1998strategy,schummer2002strategy}, where it has been the primary domain with structured preferences that allowed one to escape strong impossibility results such as the Gibbard-Satterthwaite Theorem~\cite{Gibbard77,Satterthwaite75}, and then later in algorithmic mechanism design literature~\cite{procaccia2013approximate, FotakisT14, FotakisT16, SerafinoV16, walsh2020strategy, agrawal2022learning, gravin2025approximation, ijcai2025p32} where the setting often served as an important domain for testing new concepts in approximate mechanism design, such as truthful mechanisms without monetary transfers~\cite{procaccia2013approximate}.

The computer science community has been concerned with establishing approximate efficiency guarantees and extending original strategic facility location setting in multiple directions (see~\cite{chan2021mechanism} for a recent survey on the topic), including randomization ~\cite{AlonFPT10, LuSWZ10,FotakisT10,FeldmanW13, balkanski2024randomized, barak2024mac, ijcai2025p32}; and 
opening of $k\ge 2$ facilities~\cite{procaccia2013approximate,lu2009,LuSWZ10,FotakisT10,EscoffierGTPS11,FotakisT16,walsh2020strategy,barak2024mac}.

However, despite active research in the last sixteen years since the seminal work of~\citep{ProccacciaTennenholtz2009,procaccia2013approximate}, we do not have yet a separation between random and deterministic mechanisms for the simple setting of single facility in multi-dimensional Euclidean space. This is the question we address in our work. Before listing our results, let us first do a quick overview of what is known about the problem.

On the real line $\R$, the problem is simple --- returning the median of the agent reported locations is both strategyproof and optimal \citep{moulin1980strategy,ProccacciaTennenholtz2009,procaccia2013approximate}. This is no longer the case in higher dimensions - returning the optimal solution is no longer strategyproof \citep{el2023strategyproofness}.
In $\R^d$, the simple deterministic $\cwmed$ mechanism of returning the median in each coordinate (this is known as the \emph{coordinate-wise median}, see \cref{def:cwmed}) is attractive because it is strategyproof and computationally efficient. Yet, deterministic mechanisms face a tight performance barrier: In $\R^2$, $\cwmed$ achieves an approximation ratio of $\sqrt{2}$, which is optimal among deterministic strategyproof mechanisms (\cite{durocher2009projection, goel2023optimality}). This naturally raises the question:
\begin{question}\label{question:random-beat-deterministic}
    Can randomized mechanisms strictly beat the deterministic barrier while preserving strategyproofness?
\end{question}

 This paper resolves the question above in the affirmative for $\mathbb{R}^2$ via a simple, black-box randomization of $\cwmed$: rotate the point set by a random orthogonal transformation, apply $\cwmed$ in the rotated axes, and rotate back. The resulting mechanism, $\rrcwmed$, preserves strategyproofness, yet beats the $\sqrt{2}$ barrier in expectation.

Another approach to beyond–worst-case analysis that has gained significant traction in recent years is the \emph{learning-augmented} setting, in which the mechanism receives predictions alongside the input (see \cref{sec:related-work} for more background on the learning-augmented setting). The objective is to leverage accurate predictions for improved performance while maintaining robust guarantees even when predictions are highly inaccurate.

We consider both the \emph{output-prediction} model \cite{agrawal2022learning} — where there is a single prediction for the optimal facility location and the \emph{MAC input-predictions} model \cite{barak2024mac} where the predictions are for the locations of the agents. We answer \cref{question:random-beat-deterministic} in the sense of the learning-augmented setting as well. That is, we show a better expected approximation ratio for $\R^2$ may be obtained in the setting with machine learned predictions via introducing randomization.

\subsection{Our contribution}

\subsubsection{Our results.}
We have three main results: (1) Breaking the $\sqrt{2}$ deterministic barrier (2) Doing the same in the learning-augmented setting (3) Showing that our mechanism has a better expected approximation guarantee than any GRD mechanism.

\begin{enumerate}
\item \textbf{Breaking the deterministic barrier via random rotations.}
We introduce the $\rrcwmed$ mechanism, which randomly rotates all points, applies $\cwmed$ on the rotated dataset and rotates the resulting point back. The mechanism is universally strategyproof\footnote{Universally strategyproof means it is impossible to gain from manipulating one's reported
location, regardless of the mechanism's coin tosses.} (see \cref{rem:rrcwmed-sp}) and runs in linear time (\cref{thm:sp}). In $\R^2$, we prove the expected approximation ratio of $\rrcwmed$ is \emph{exactly} $\tfrac{4}{\pi}\approx1.273$ (\cref{thm:alg-R2}), strictly improving on the tight $\sqrt{2}\approx1.414$ achievable by any deterministic strategyproof mechanism.

In $\R^d$, our mechanism has a polynomial runtime ($O(d^3 + nd^2)$) and we show its expected approximation ratio lies in $\brk*{\sqrt{2}-O(d^{-1/2}), 1.547}$ (\cref{thm:rrcwmed-expected-approximation-ratio-Rd}). We conjecture that our lower bound of $\sqrt{2}$ is tight (up to $O(d^{-1/2}))$.

\item \textbf{Improved learning-augmented mechanisms.} We show that our random rotation method may also improve expected approximation ratios in learning-augmented settings.

For the \emph{output-prediction} model, \citet{agrawal2022learning} proposed $\cmp$ mechanism (adding $\lfloor c n\rfloor$ copies of the prediction $\hat{o}$ for $c \in (0,1)$, and applying $\cwmed$), attaining consistency of $\frac{\sqrt{2c^2+2}}{1+c}$ with robustness of $\frac{\sqrt{2c^2+2}}{1-c}$. We show that $\rrcmp$ mechanism, obtained by swapping $\cwmed$ with $\rrcwmed$ in $\cmp$, has an improved consistency-robustness trade-off (\cref{thm:rrcmp-approximation}): it is $\min\crl*{\frac{\sqrt{2c^2 +2}}{1+c}, \frac{4}{\pi}}$-consistent, $\min\crl*{\frac{\sqrt{2c^2+2}}{1-c}, \frac{4}{\pi} \frac{1+c}{1-c}}$-robust and has a smooth degradation between consistency and robustness as a function of the prediction error.

This trade-off is better for many values of $c$: a calculation shows there is strict improvement for $c \in (0, 0.118], \eta \ge 0$ or $c \in [0.118, 0.785), \eta > \frac{4(1+c)-\pi\sqrt{2c^2+2}}{(\pi-4c)(1+c)}$. Since \citet{agrawal2022learning} show that the $\cmp$ mechanism has optimal consistency-robustness trade-off out of all deterministic anonymous strategyproof mechanisms, our result implies breaking the deterministic bound in this setting as well.

For the \emph{MAC input-predictions} model, by replacing $\cwmed$ with $\rrcwmed$ in the ``use predictions-only or use $\cwmed$'' template of \cite{barak2024mac} improves the achieved bound: Our result improves the ``worst-case arm'' from $\sqrt{2}$ to $\tfrac{4}{\pi}$, implying an approximation ratio of $\min \crl*{1 + \frac{4 \delta}{1-2\delta},\tfrac{4}{\pi}}$ rather than $\min \crl*{1 + \frac{4 \delta}{1-2\delta},\sqrt{2}}$.

\item \textbf{Limits of GRD mechanisms.}
We show the following lower bounds on GRD (Generalized Random Dictators) mechanisms --- mechanisms which randomly choose an agent-reported location (see \cref{def:GRD}), a very wide-spread class of mechanisms (see \cref{sec:related-work} for examples). For $\R^2$: we show that any GRD mechanism has expected ratio at least $\tfrac{4}{\pi}$ (\cref{thm:lb-rand-dict-R2}). For $\R^d$, we show a lower bound of $\sqrt{2}-O(1/d)$ (\cref{thm:lb-rand-dict-Rd}). This means our $\rrcwmed$ (\cref{alg:rrcwmed}) is (weakly) better than any GRD mechanism in $\R^2$.

\end{enumerate}

\subsubsection{Our techniques.}

Our key lever for improving performance in expectation is via introducing random rotation matrices into the mechanism’s pipeline (\cref{alg:rrcwmed,alg:rrcmp}). Because the rotation matrix is sampled independently of the agents’ reports, strategyproofness is preserved in the strongest sense (see \cref{rem:rrcwmed-sp}). The challenge lies in the expected approximation ratio analysis under this randomization.

\paragraph{Upper bound techniques.}
For the \(\R^2\) upper bound for \(\rrcwmed\) (\cref{lem:R2-ub}), we exploit two elementary facts: (i) Euclidean (\(\ell_2\)) pairwise distances are invariant under rotations, so we may analyze the rotated dataset instead of the original one; and (ii) the coordinatewise median minimizes the \(\ell_1\) objective, allowing us to convert a bound on the mechanism’s expected \(\ell_2\)-cost into a bound in terms of the optimal \(\ell_1\)-cost.\footnote{That is, \(\scost_1\); see \cref{def:scost}.} The bridge between \(\ell_2\) and \(\ell_1\) is provided by \cref{lem:expected-abs-value-of-each-rotated-component}, which shows that for any fixed \(v\in\R^2\),
$\E_\theta\brk*{\norm{R_\theta v}_1} \;=\; \frac{4}{\pi}\,\norm{v}_2$,
where \(R_\theta\) is a uniformly chosen random rotation matrix. Combining these parts together yields the stated upper bound.

For the learning-augmented setting of \(\rrcmp\) with an output-prediction model \cite{agrawal2022learning}, the analysis parallels the above with additional technical ingredients. An interesting ingredient is  \cref{lem:geometric-median-approx-robustness} which establishes robustness of the social cost minimizer solution ($1$-median) to an \emph{insertion corruption}—adding a \(c\in(0,1)\) fraction of arbitrary points. This yields an improved bound compared to \citet{barak2024mac}, as we analyze ``insertion-only'' corruption rather than ``edit corruption''. Interestingly, although \citet{barak2024mac} introduce this type of analysis for the sake of MAC input predictions model, we use our improved robust statistics bounds for the output-prediction model of \citet{agrawal2022learning}.

\paragraph{Lower bound techniques.}
For \(\R^2\) (\cref{lem:rand-cwmed-lb}), we design a “two clusters and an outlier” instance: two nearby clusters \(A,B\) of size \(k\gg 1\) each, and a far outlier \(C\) (see \cref{fig:rotate-2d-illustration}). When the clusters are sufficiently large, the optimal solution is close to the median of \(A\cup B\). In contrast, \(\rrcwmed\) is “pulled” by the outlier so that after a random rotation it takes one coordinate from \(R_\theta A\) and the other from \(R_\theta B\), incurring a larger cost. By placing \(A,B,C\) appropriately, this behavior persists for all \(\theta\) outside a measure-zero set, giving the \(\tfrac{4}{\pi}\) lower bound.

In $\R^d$, the construction is analogous, but the analysis calls for heavier tools.
We invoke results from random matrix theory (e.g. as $d\to\infty$,
suitably rescaled coordinates of random rotations behave approximately as i.i.d.\ Gaussians
\cite{diaconis1987dozen,meckes2019random}) together with high-dimensional probability (e.g. concentration of measure
for Lipschitz functions of random sub-Gaussian variables on $SO(d)$).
The proof therefore entails sharper quantitative estimates (e.g., bounding expectations of Gaussian functions),
but the underlying idea remains the same as in the planar case: we show that, with high probability over the random
rotation, each rotated cluster point stays at distance about $1$ from the $\rrcwmed$ output, yielding the desired
expected-cost lower bound.

Finally, for GRD mechanisms we consider \emph{rotation-invariant instances} whose points lie on the unit sphere. Any GRD rule must select one of the sphere points and thus pay a comparatively large cost, whereas the optimum can place the facility at the origin and pay \(n\). In \(\R^2\) (\cref{thm:lb-rand-dict-R2}), evenly spaced points on the unit circle yield cost approaching \(\tfrac{4}{\pi}n\). In \(\R^d\) (\cref{thm:lb-rand-dict-Rd}), a random-instance construction combined with the probabilistic method \cite{alon2016probabilistic} shows that, with high probability, any GRD mechanism incurs expected cost \(\approx \sqrt{2}\,n\), establishing the desired lower bound.

\subsection{Further Related Work}\label{sec:related-work}

Beyond the works already discussed so far, we highlight several additional lines of related work.

\paragraph{Utilitarian facility location mechanism design.}

In the plane \(\R^2\), the tight approximation ratio is known to be \(\sqrt{2}\): upper bounds are given by \citet{bespamyatnikh2000mobile,meir2019strategyproof} and a matching lower bound by \citet{goel2023optimality}.
In higher dimensions \(\R^d\), the approximation ratio was known to be upper bounded by $\sqrt{d}$ \cite{meir2019strategyproof}, and later improved to $1.547$ \cite{gravin2025approximation}.

\paragraph{Learning-augmented facility location mechanism design.}
Learning augmented algorithm design is a popular field of study in recent years within the paradigm of beyond the worst case analysis of algorithms; see \citep{DBLP:journals/cacm/MitzenmacherV22} for a survey and \citep{ALPSwebsite} for a curated, frequently updated bibliography.
Learning-augmented \emph{mechanism design} was initiated by \citet{agrawal2022learning,ijcai2022p81} and studied further for facility location
\citep{agrawal2022learning,ijcai2022p81,istrate2022mechanism,chen2024strategic,balkanski2024randomized,barak2024mac,Shi2025PredictionAugmented,walsh2025mechanism,gravin2025approximation}
and in other domains such as distortion \citep{Berger2024Distortion,FilosRatsikas2025Utilitarian},
auctions \citep{Medina2017Revenue,ijcai2022p81,Balkanski2024Online,Caragiannis2024Randomized,Lu2024Competitive,Gkatzelis2025Clock,Caragiannis2025Mechanisms},
strategic scheduling \citep{ijcai2022p81,Balkanski2023Strategyproof,ChristodoulouSV24},
and equilibrium analysis \citep{GkatzelisKST22,Istrate2024Equilibria}.

In the context of this work, the most relevant work is of \cite{agrawal2022learning} studying the output-prediction model, and \cite{barak2024mac} studying the $\mac$ input predictions model.
In the output-prediction model, \citet{agrawal2022learning} study the setting of $\R^2$, where the prediction consists of a single value representing advice for where to locate the facility.
Given a trust parameter $c \in [0,1]$, their mechanism, $\cmp$, guarantees an approximation ratio of at most $\smash{\frac{\sqrt{2c^2 + 2}}{1+c}}$ if the prediction is perfect, and $\smash{\frac{\sqrt{2c^2 + 2}}{1-c}}$ in the worst case, and a smooth degradation between consistency and robustness as a function of the prediction error $\eta$. The $\cmp$ mechanism operates by adding copies of $\hat{g}$ to the dataset and applies $\cwmed$ on the resulting dataset. They also show the optimality of $\cmp$: for any $c \in (0,1)$ and a $\frac{\sqrt{2c^2 + 2}}{1+c}$-consistent deterministic anonymous\footnote{A facility location mechanism \( f : \mathbb{R}^n \to \mathbb{R} \) is \emph{anonymous} if for every location profile
\((x_1,\dots,x_n) \in \mathbb{R}^n\) and for every permutation \( \pi \) of \( \{1,\dots,n\} \), $f(x_1,\dots,x_n) = f(x_{\pi(1)},\dots,x_{\pi(n)})$.}  strategyproof mechanism, it must have robustness no better than $\frac{\sqrt{2c^2+2}}{1-c}$. We show that our $\rrcmp$ breaks the deterministic barrier --- it has an improved consistency-robustness trade-off. \citet{gravin2025approximation} extend the analysis of $\cwmed$ and $\cmp$ mechanisms to $\R^d$. By analyzing a factor revealing optimization problem, they show an upper bound of $1.547$ on the approximation ratio of $\cwmed$. In the $\mac$ input predictions model, \citet{barak2024mac} show a $\min\crl*{1 + \frac{4\delta}{1-2\delta}, \sqrt{2}}$ approximation ratio given predictions for the $n$ agent locations where a $1-\delta$ fraction of the predictions are correct.

\paragraph{Improved approximation via rotation.}
\citet{Gershkov2019} showed that on certain distributions (e.g., identical marginals satisfying additional regularity) a fixed \(\pi/4\) rotation followed by the coordinate-wise median can outperform no rotation. However, in the worst case the \(\sqrt{2}\) bound remains tight: the lower bound of \citet{goel2023optimality} applies after undoing the fixed rotation. Our result obtains a strictly better \emph{expected} approximation (\(\approx 1.27\)) for \emph{all} instances by sampling a \emph{uniform random} rotation, rather than committing to a deterministic one.
\citet{meir2019strategyproof} and \citet{goel2023optimality} conjectured that random rotation may yield better results. \citet{meir2019strategyproof} focus on $n=3$ and write:
``Note that by using randomization, we are likely to get a certain improvement.
For example, we can select the axes according to a random rotation, and then run
the multi-median mechanism. We leave the analysis of such mechanisms to future
work.''
\citet{goel2023optimality} write: ``The question of how well a randomized mechanism might approximate
the social cost of the geometric median remains open. A potentially good candidate is the
mechanism that chooses a coordinate-wise median after a uniform rotation of the orthogonal
axes. While its analysis seems hard in general, finding its AR on the worst-case profile might give a useful lower bound.''

This work closes this open problem, by pinning down the exact expected approximation of $\rrcwmed$ for the general setting of any $n$ agents instance.

\paragraph{GRD mechanisms.}
A prominent class of randomized rules in facility location is \emph{GRD} (Generalized Random Dictators, also known as peaks-only; see \cref{def:GRD}), which choose a reported location with some probability. GRD rules are appealing for strategic reasons, but their approximation power is limited; our lower bounds quantify these limits and imply that no GRD mechanism has a better expected approximation ratio than \(\rrcwmed\) on \(\R^2\). Examples of GRD mechanisms include Uniform Random Dictator (choosing each reported location with probability \(1/n\)) \citep{AlonFPT10,alon2010walking,Feldman16}, Proportional \citep{LuSWZ10,barak2024mac}, Inverse-Proportional \citep{escoffier2011strategy}, the median mechanism \citep{black1948rationale,moulin1980strategy,procaccia2013approximate}, one-dimensional Percentile mechanisms \citep{sui2013analysis}, Phantom mechanisms \citep{moulin1980strategy,masso2011strategy}, Random Rank \citep{AzizLSW22}, “allocate to a random agent of the optimal cluster’’ \citep{fotakis2021strategyproof}, PCD \citep{meir2019strategyproof}, RD+PCD \citep{ijcai2025p449}, and \(q\)-QCD \citep{meir2019strategyproof}.

\paragraph{Other work using multi-dimensional median generalizations.}
The $\projmed$ is a generalization of the median to multiple dimensions which is also known to have an approximation ratio of $\frac{4}{\pi}$ for the social cost objective \cite{durocher2009projection}. However, as we discuss in \cref{sec:proj-median-connection}, \(\projmed\) is not strategyproof (\cref{rem:projmed-not-sp}); we also explain how its approximation guarantee relates to the expected approximation of \(\rrcwmed\).
\citet{zampetakis2023bayesian} study Bayesian mechanism design. They use the Lugosi-Mendelson median, a different multi-dimensional median extension from heavy-tailed statistics, to achieve asymptotical approximately optimal cost under distributional assumptions.

\section{Conclusion and Future directions}\label{sec:conclusions-and-future-dir}

This work demonstrates that randomized rotation is a principled way to surpass the deterministic barrier in utilitarian facility location, both in the classical and learning-augmented settings, and that GRD (peaks-only) mechanisms are not the right tool for the problem.

We list a few of the many remaining avenues for future work. One such potential avenue is to try and pin down the exact approximation ratio of random mechanisms in $\R^2$ by either coming up with a new lower bound or alternatively beat the $\frac{4}{\pi}$ result of $\rrcwmed$ by coming up with a new mechanism. For the upper bound direction, our lower bound of GRD mechanisms implies that new ideas are required.

Another direction would be to prove \cref{conj:rrcwmed-ub-Rd}, and show that indeed \cref{alg:rrcwmed} has an expected approximation ratio of at most $\sqrt{2}$ for any $d \in \N$.

Other future directions may also include studying other objectives, such as the egalitarian cost, where the objective of the mechanism designer is to minimize the maximum agent cost, rather than the sum (where the deterministic barrier of approximation ratio is known to be $2$). It is an interesting question to determine whether it is possible to break the deterministic barrier for the maximum cost objective.

We think it is likely that our random rotation method may be utilized to get improvements for other facility location problems as well.
\subsection{Roadmap}
Section~\ref{sec:preliminaries} sets up notation and preliminaries. Section~\ref{sec:rand-mech} introduces $\rrcwmed$ and proves its expected approximation ratio guarantees: exactly $\tfrac{4}{\pi}$ in $\mathbb{R}^2$, and the bounds of $\brk*{1.41,1.547}$ in $\mathbb{R}^d$. Section~\ref{sec:predictions-learning-aug} gives the results for the learning-augmented settings - improved guarantees via randomness. Section~\ref{sec:lower-bounds} contains our GRD (Generalized Random Dictator) lower bounds.

\section{Preliminaries}\label{sec:preliminaries}
We start by giving the formal problem definition. The social cost is the objective function which we want to minimize.

\begin{definition}[Social Cost]\label{def:scost}
    The social cost of a dataset $P \subset \R^d$ and a center $m \in \R^d$ is the sum of $l_q$ norm distances from $P$ to $m$. That is:
    \[
        \scost_q(P,m) = \sum_{p_i \in P} \norm{p_i - m}_q.
    \]
    In this work, our main focus is on euclidean distances ($q = 2$), and thus we use the notation: \ $\scost(P,m) := \scost_2(P,m)$.
\end{definition}

\paragraph{The Utilitarian facility location mechanism design problem.}
As input the mechanism receives $P$, $n$ locations of strategic agents in Euclidean space $\R^d$, for $n,d \in \N$. The output of the mechanism is a location in $\R^d$ for the facility. The cost of each agent is the $l_2$ distance from the agent location to the returned facility.
The \emph{expected approximation ratio} $\alpha(M)$ of a random mechanism $M$ is defined as the worst case ratio between the expected cost of the mechanism and that of the optimal solution. Formally:
\[
    \alpha(M) := \sup_{P \subset \R^d} \frac{\E\brk*{\scost(P,M(P))}}{\scost\prn*{P,\prn*{m(P)}}},
\]
where $m(P) = \argmin_{m \in \R^d} \crl*{\scost(P,m)}$ is the optimal solution.
The goal is to design a strategyproof mechanism $M$ that returns a location that minimizes the expected approximation ratio, where strategyproofness is defined as follows:

\begin{definition}[Strategyproofness]\label{def:strategyproofness}
    Let $P = (p_1, \ldots ,p_n) \subset \R^d$ be the locations of $n$ agents in space, and for any $i \in [n]$ let $(P_{-i},y)$ be $P$ where $p_i$ is replaced by $y \in \R^d$.
    A mechanism $M$ is strategyproof iff for any agent $i \in [n]$ located in location $p_i$, the cost of the agent may not decrease by reporting a location $p'_i \neq p_i$. That is: $d\prn*{p_i, M(P)} \le d\prn*{p_i, M\prn*{P_{-i},p'_i}}$. We demand the strategyproofness condition for any realization of the coin tosses of the mechanism. This is also known as universal truthfulness.
\end{definition}

We introduce the following convention which we use throughout:
For any vector $v \in \R^2$ let $v_1$, $v_2$ denote the $x$, $y$ coordinates of $v$, respectively. For a multiset of vectors $V$ in $\R^2$ and $j \in [2]$, let $V_j = \crl*{\crl*{v_j \mid v \in V}}$ be the multiset of the $j$'th coordinates of the vectors in $V$.

\paragraph{Generalizations of the median to $\R^d$.}
There are multiple different generalizations of the median to higher dimensions. We list three of them; the first is the $\geomed$:
\begin{definition}[$\geomed$]\label{def:geo-med}
    For a multi-set of $n$ points $P \subseteq \R^d$, let 
    
    $\geomed(P) := \argmin_{m \in \R^d} \crl*{\scost_q(P,m)}$ be the social cost minimizer solution for the given dataset $P$.
\end{definition}
The geometric median is also known by many other names such as the Fermat point, the Euclidean median, the spatial median, and the $1$-median (as a special case of $k$-medians for $k=1$). 

A different median generalization is the $\cwmed$, which is the median of the given dataset projected to each coordinate separately:
\begin{definition}[$\cwmed$ - Coordinate-Wise Median]\label{def:cwmed}
    For a multi-set of $n$ points $P \subseteq \R^d$:
    
    $\cwmed(P) := (l_1, \ldots, l_d)$ where $l_j$ is the median of $P_j$ for all $j \in [d]$.
\end{definition}

In the following proposition, we list known facts about $\cwmed$:
\begin{prop}\label{prop:cwmed-properties}
    $\cwmed$ has the following properties: (1) it is strategyproof \cite{procaccia2013approximate}, (2) it is the $\scost_1$ minimizer solution: $\cwmed(P) \in \argmin_{x \in \R^d} \crl*{\scost_1(P,x)}$. \cite{meir2019strategyproof} (3) its approximation ratio in $\R^2$ is exactly $\sqrt{2}$ \cite{BeregBKS06, durocher2009projection, meir2019strategyproof} and it is the optimal mechanism in the family of deterministic strategyproof mechanisms \cite{goel2023optimality} and (4) its approximation ratio in $\R^d$ is upper bounded by $\sqrt{6 \sqrt{3} - 8} \approx 1.547$, and is tight for large $d$ up to $O\prn*{\frac{1}{d}}$ \cite{gravin2025approximation}.
\end{prop}

\emph{Rotations in Euclidean space.}
Next, we give the preliminaries on random rotation matrices in Euclidean space.

\begin{definition}[Rotation Matrix]\label{def:rotation-matrix}
    A \emph{rotation matrix} in $\mathbb{R}^d$ is an orthogonal matrix $R \in \mathbb{R}^{d \times d}$ with determinant $\det(R)=1$. Equivalently,
    \[
    R^\top R = I_d \quad \text{and} \quad \det(R)=1,
    \]
    where $I_d$ is the $d \times d$ identity matrix. The set of all such matrices forms the \emph{special orthogonal group}
    \[
    SO(d) := \{ R \in \mathbb{R}^{d \times d} \;:\; R^\top R = I_d,\ \det(R)=1 \}.
    \]

    In $\R^2$, the two dimensional \emph{rotation matrix} by angle $\theta \in [0,2\pi)$ is defined as
    \[
        R_\theta \;=\;
        \begin{pmatrix}
            \cos\theta & -\sin\theta \\
            \sin\theta & \phantom{-}\cos\theta
        \end{pmatrix}.
    \]
\end{definition}

The $\rcwmed$ is obtained by rotating the dataset, computing the $\cwmed$ of the rotated dataset, and rotating it back. Formally:
\begin{definition}[$\rcwmed$]\label{def:rcwmed}
    Let $P$ be a multi-set of $n$ points $P \subseteq \R^d$, and let $R \in SO(d)$ be a rotation matrix in $\R^d$. Then $\rcwmed$ is defined to be the $\cwmed$ of the rotated multi-set $RP := \crl*{Rp \mid p \in P}$. \[\rcwmed(P,R) := R^{-1} (\cwmed(RP)).\]
    
    That is, $\rcwmed$ is obtained by rotating the dataset, computing its $\cwmed$ and rotating back.
\end{definition}

\paragraph{Learning augmented settings}

In this work we study two learning augmented prediction settings:
\begin{itemize}
    \item \emph{Output prediction -} In the setting of \citet{agrawal2022learning}, there is a single prediction $\hat{g} \in \R^d$ for the optimal facility location $g$ for agent locations $P$. The error parameter is $\eta = \frac{\norm{\hat{g} - g}}{\nf{\scost\prn*{P,g}}{n}}$; the distance between the predicted optimal location $\hat{g}$ and the true optimal location $g$, normalized by the average optimal social cost\footnote{\citet{agrawal2022learning} define the social cost to be the average cost: $\scost'(P,t) = \nf{\sum_{p \in P} \norm{p-t}_2}{n}$. Since there's only a multiplicative factor of $n$, the problems are equivalent.}: $\nf{\scost\prn*{P,g}}{n}$. \emph{Consistency and robustness}: In this setting, a mechanism is said to be $\gamma$-consistent (or to exhibit consistency level $\gamma$) if its expected approximation ratio does not exceed $\gamma$ when $\eta = 0$. Likewise, a mechanism is $\beta$-robust (or exhibits robustness level $\beta$) if its expected approximation ratio does not exceed $\beta$ for any value of $\eta$. In general, a consistency–robustness trade-off is observed, wherein achieving higher consistency (smaller $\gamma$) typically entails diminished robustness (larger $\beta$).

    \item \emph{$\mac$ input predictions -} In the setting of \citet{barak2024mac}, the predictions are for the agent locations, and are assumed to be $\mac$ predictions --- most predictions (at least a $1-\delta$ fraction for $\delta \in (0,0.5)$) are assumed to be correct up to a small negligible $\eps$ additive error.
\end{itemize}

\section{Breaking The Deterministic Barrier Via Random Rotations}\label{sec:rand-mech}

In this section we give a random strategyproof mechanism for the facility location mechanism design problem. For $\R^2$ we pin down the exact approximation ratio: it is exactly $\frac{4}{\pi}$. In $\R^d$ we show the approximation ratio lies in $[1.41,1.55]$. Our $\frac{4}{\pi} \approx 1.27$ approximation result we get for $\R^2$ strictly beats the lower bound on the best deterministic mechanism of $\sqrt{2} \approx 1.41$, hence it demonstrates the separation between random and deterministic strategyproof mechanisms for this problem.
We now give our mechanism, $\rrcwmed$.

\begin{algorithm}[H]
\caption{$\rrcwmed(P)$: Randomly Rotated Coordinate-wise Median}
\label{alg:rrcwmed}
\begin{algorithmic}[1]
\Require $P = \{p_1,\dots,p_n\} \subseteq \mathbb{R}^d$
\State Sample a uniformly random rotation matrix $R \in SO(d)$. \Comment{Haar uniform\footnotemark}
\State Compute the rotated dataset $P' \gets \{p'_i = R p_i : p_i \in P\}$.
\State Compute the coordinate-wise median $h' \gets \cwmed(P')$.
\State \Return $h \gets R^{-1} h'$. \Comment{$R^{-1}=R^\top$ for $R\in SO(d)$}
\end{algorithmic}
\end{algorithm}
\footnotetext{A random matrix $R \in SO(d)$ is said to be Haar-uniform if its distribution is the unique probability measure on $SO(d)$ that is invariant under left and right multiplication by any fixed rotation \cite{Vershynin2025HDP,mezzadri2006generate}.}

\paragraph{Intuition.} Before showing the guarantees of the mechanism, we give some intuition.

\begin{figure}[t]
    \centering
    \includegraphics[width=0.75\linewidth]{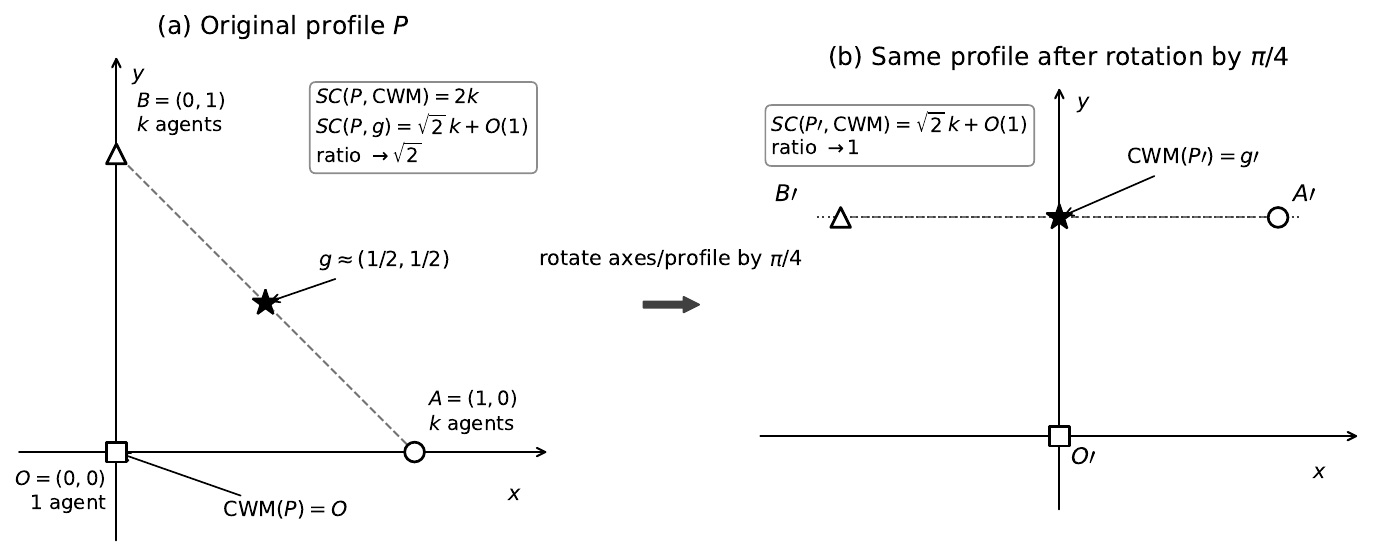}
    \caption{
    Illustration of the axis-dependence of $\cwmed$ under the Euclidean objective. 
    On the left, there's an instance with $k \gg 1$ points at $(1,0)$, $k$ points at $(0,1)$, 
    and one point at the origin $(0,0)$. On the right, we get the instance obtained by rotating the left instance counterclockwise by $\pi/4$.
    }
    \label{fig:rr-cwm-rotation-intuition}
\end{figure}

$\cwmed$ is a natural deterministic primitive, but its weakness is that it is tied to a fixed (arbitrary) coordinate system. To illustrate this, consider the instance shown in Figure~\ref{fig:rr-cwm-rotation-intuition}: there are $k \gg 1$ points at $(1,0)$, $k$ points at $(0,1)$, and one point at the origin $(0,0)$. One can place the facility at $(1/2,1/2)$ and incur cost $\sqrt{2}k + O(1)$, but CWM returns the origin and incurs cost $2k$, which implies an approximation ratio of $\sqrt{2} - O(1/k)$. On the other hand, if the points are rotated counterclockwise by $\pi/4$, the Euclidean instance is unchanged up to isometry: all pairwise Euclidean distances remain the same. However, $\cwmed$’s cost drastically improves, becoming $\sqrt{2}k + O(1)$ as well.

CWM is strategyproof and minimizes the $\ell_1$ social cost in the chosen axes. The Euclidean objective, however, is rotationally invariant. Thus, the deterministic barrier comes from an axis-alignment mismatch: for a fixed coordinate system, an adversary can orient the instance so that the $\ell_1$-based $\cwmed$ solution is poorly aligned with the Euclidean optimum.

$\rrcwmed$ removes this axis dependence by choosing the coordinate system uniformly at random. For every realized rotation, we still run an ordinary $\cwmed$ in the rotated coordinates, so strategyproofness is preserved in the strong universal sense. The inverse rotation is only a change of coordinates back to the original Euclidean space, and Euclidean distances are unchanged by this operation.

So the transformation is natural for three reasons: it preserves the truthfulness of $\cwmed$, it respects the rotational symmetry of the Euclidean objective, and it replaces the worst-case $\ell_1$ versus Euclidean norm distortion of a fixed coordinate system by its average distortion. A fixed rotation would not achieve this, because the adversarial instance could simply be rotated accordingly; the uniform random rotation eliminates any preferred direction.

\paragraph{The guarantees of $\rrcwmed$.}
We begin by noting the strategyproofness of the mechanism.
\begin{thm}\label{thm:sp}
     \cref{alg:rrcwmed} is strategyproof and runs in time $O(d^3 + n d^2)$.\footnote{Consequently, for any fixed dimension (e.g., $d=2$), \cref{alg:rrcwmed} runs in linear time $O(n)$.}
\end{thm}
\begin{proof}
    The strategyproofness of \cref{alg:rrcwmed} follows from the strategyproofness of $\cwmed$ (\cref{prop:cwmed-properties}), since the agents have no control over the random rotation matrix the algorithm chooses.

    Computing \cref{alg:rrcwmed} can be done in $O(d^3 + nd^2)$ as the rotation matrix may be sampled in $O(d^3)$ (see \cite{genz2000methods, mezzadri2006generate}), the matrix-vector multiplication takes $O(n d^2)$ and median selection takes $O(nd)$ time ($O(n)$ time for median computation in each of the $d$ coordinates).
\end{proof}

\begin{remark}\label{rem:rrcwmed-sp}
    Usually in random mechanism design, we only require the mechanism to be strategyproof in expectation (telling the truth yields the best expected utility for the agents). In our case the strategyproofness holds in the stronger universal deterministic sense: telling the truth is a deterministic dominant strategy for any realization of the mechanism random coin tosses. Note that this does not hold for the maximum cost objective facility location mechanism design problem even on the line ($\R$) \cite{branzei2015verifiably}.
\end{remark}

We now move on to analyzing the approximation ratio of our mechanisms.
We start with $\R^2$, and then move to $\R^d$.
\subsection{$\R^2$ approximation analysis}

In $\R^2$, the rotation matrix has a simple form, as explained by the following remark:
\begin{remark}\label{rem:R2-mechanism-form}
    For $\R^2$, and $k \in \N$, \cref{alg:rrcwmed} is equivalent to sampling $\theta \sim \mathrm{Unif}\crl*{[0, k\frac{\pi}{2})}$ and returning $\rcwmed(P,R_\theta)$. In our analysis we use $k=1$: $\theta \sim \mathrm{Unif}\crl*{[0, \frac{\pi}{2})}$.
\end{remark}
\begin{proof}
    Sampling a random matrix $R \in SO(2)$ is equivalent to sampling a random $2 \times 2$ matrix $R_\theta$ where $\theta \sim U\crl*{[0,2\pi)}$ \cite{meckes2019random}. Also, $\rcwmed$ is $\frac{\pi}{2}$-periodic, as any $\frac{\pi}{2}$ rotation does not change the location of the $\cwmed$: while the names of the axes may change, the axes lines remain the same. Thus, the median in each axis remains the same as well.
\end{proof}

We now pin down the exact expected approximation ratio of our mechanism:

\begin{thm}\label{thm:alg-R2}
    The approximation ratio of \cref{alg:rrcwmed} for $\R^2$ is exactly $\alpha(\rrcwmed) = \frac{4}{\pi} \approx 1.273$.
\end{thm}
To show \cref{thm:alg-R2}, we give the tight upper bound (\cref{lem:R2-ub}) and the matching lower bound (\cref{lem:rand-cwmed-lb}).

We begin with \cref{lem:R2-ub}, where we show a more general upper bound for any $l_q$ norm. Plugging in $q = 2$ gives us the upper bound for \cref{thm:alg-R2}.

\begin{lem}\label{lem:R2-ub}
    For any $q \ge 2$, \cref{alg:rrcwmed} has an approximation ratio of at most $\prn*{2^{\nf{1}{2} - \nf{1}{q}}} \frac{4}{\pi}$ for the $\scost_q$ cost function\footnote{For $q = 1$, $\cwmed$ mechanism is already optimal \cite{meir2019strategyproof} and achieves an approximation ratio of $1$.}. Equivalently:
    \[
        \E\brk*{\scost_q(P,\rrcwmed(P))} \le \frac{4}{\pi} \cdot \prn*{2^{\nf{1}{2} - \nf{1}{q}}} OPT,
    \]
    where $OPT = \scost_q(P,\geomed(P))$.
\end{lem}

Before showing the upper bound, we first show a useful lemma, which we also use in the learning augmented setting (\cref{sec:predictions-learning-aug}):
\begin{lem}\label{lem:expected-abs-value-of-each-rotated-component}
    Let $v \in \R^2$ be a fixed vector and let $\theta \sim \mathrm{Unif}\crl*{[0,\frac{\pi}{2})}$. Then 
    \[
        \E\brk*{\norm{R_\theta v}_1} = \frac{4}{\pi} \norm{v}_2.
    \]
    
\end{lem}
\begin{proof}

It follows from the $\ell_1$ norm definition that $\norm{R_\theta v}_1 = \abs{(R_\theta v)_1} + \abs{(R_\theta v)_2}$, and thus we get the desired by showing:
\[
    \E\brk*{\abs{(R_\theta v)_1} + \abs{(R_\theta v)_2}} = \frac{4}{\pi} \norm{v}_2.
\]

Let $(y,\alpha)$ be the polar representation of $v$.
Then: \[
(R_\theta v)_1 = y (cos(\theta) cos(\alpha) - sin(\theta) sin(\alpha)) = y \ cos(\theta+\alpha),
\]
\[
(R_\theta v)_2 = y (sin(\theta) cos(\alpha) + cos(\theta) sin(\alpha)) = y \ sin(\theta+\alpha).
\]
Thus:
\begin{align*}
    \E\brk*{\abs{(R_\theta v)_1} + \abs{(R_\theta v)_2}} &= \int_{0}^{\frac{\pi}{2}} \frac{2}{\pi} y \prn*{ \abs{cos(\theta+\alpha)} + \abs{sin(\theta+\alpha)}} d\theta \\
    & = \frac{2 y}{\pi} \int_{\alpha}^{\alpha + \frac{\pi}{2}} \prn*{\abs{cos(u)} + \abs{sin(u)}} du \\
    & \stackrel{(\star)}{=} \frac{2 y}{\pi} \int_{0}^{\frac{\pi}{2}} \prn*{\abs{cos(u)} + \abs{sin(u)}} du = \frac{2 y}{\pi} \prn*{1 + 1} = \frac{4}{\pi} \ \norm{v}_2,
\end{align*}
where equality $(\star)$ holds as $f(u) = |cos(u)| + |sin(u)|$ is $\frac{\pi}{2}$-periodic: \[
f\prn*{u + \frac{\pi}{2}} = \abs*{cos\prn*{u + \frac{\pi}{2}}} + \abs*{sin\prn*{u + \frac{\pi}{2}}} = \abs{- sin(u)} + \abs{cos(u)} = f(u).
\]

\end{proof}

We now prove \cref{lem:R2-ub}.
\begin{proof}

Let $g = \geomed(P)$ and let $h$ be the location returned by the algorithm. For any $i \in [n]$: let $v_i = p_i - g$. Let $OPT = \sum_{i \in [n]} \norm{p_i - g}_q = \sum_{i \in [n]} \norm{v_i}_q$.

For any vector $v \in \R^2$ let $v_1$, $v_2$ denote the $x$, $y$ coordinates of $v$, respectively.

For any $j \in [2]$: let $(h')_j = median(\crl*{(R p_i)_j})$ be the median of the multiset of $j$ coordinates of $P'$ (for $j=1$ it is the $x$ axis and for $j=2$ it is the $y$ axis obtained after a rotating by $\theta$).

$ALG_\theta$, the cost of the algorithm given $\theta$ is:
\begin{align*}
    ALG_{\theta} & = \sum_{i \in [n]} \norm{p_i - h}_q \le \sum_{i \in [n]} \norm{p_i - h}_2 = \sum_{i \in [n]} \norm{R_\theta p_i - R_\theta h}_2 \\
    & \le \sum_{i \in [n]} \norm{R_\theta p_i - h'}_1 \le \sum_{i \in [n]} \norm{R_\theta p_i - R_\theta g}_1 = \sum_{i \in [n]} \norm{R_\theta v_i}_1, \numberthis \label{eq:alg-bound-each-comp-by-theta}
\end{align*}
where the first inequality holds as $q \ge 2$ , the second equality is due to the fact that rotation is an isometry for $\norm{\cdot}_2$, the second inequality is due to the fact that $\norm{x}_2 \le \norm{x}_1$ for any $x \in \R^2$, and the last inequality is due to the fact that $h'_j$ is the minimizer of function $f_j(x) = \sum_{i \in [n]} \abs{a^{(j)}_i - x}$ (as $h'_j$ is the median of the projected points).

By taking the expectation (on \cref{eq:alg-bound-each-comp-by-theta}) and applying \cref{lem:expected-abs-value-of-each-rotated-component} we get 
that the expected cost of the algorithm is:
\[
        \E_\theta[ALG_\theta] \le \sum_{i \in [n]} \frac{4}{\pi} \norm{v_i}_2 \stackrel{\text{since } q \ge 2}{\le} \ \sum_{i \in [n]} \frac{4}{\pi} \prn*{2^{\nf{1}{2} - \nf{1}{q}}} \norm{v_i}_q = \frac{4}{\pi} \prn*{2^{\nf{1}{2} - \nf{1}{q}}} OPT.
\]

\end{proof}

We now show a tight lower bound, closing the gap in the approximation ratio analysis:
\begin{lem}\label{lem:rand-cwmed-lb}
    The expected approximation ratio of \cref{alg:rrcwmed} is at least $\frac{4}{\pi}$.
\end{lem}
\begin{proof}
Consider the instance $P \subset \R^2$: $k$ points in $(1,0)$, $k$ points in $(0,1)$ and $1$ point in $(-M,-M)$, for $M\gg1 \in \N$ and $k = M^2$.

Let $A = (1,0)$, $B = (0,1)$ and $C = (-M, -M)$. 
Consider the new point locations $P'$ obtained by a rotation by $\theta \in [0,\frac{\pi}{2}]$.
Let $s = sin(\theta)$ and $c = cos(\theta)$. Then the new locations of the points are $A', B', C'$ where:
$A' = R_\theta \cdot A = (c,s)$, $B' = R_\theta \cdot B = (-s,c)$ and $C' = R_\theta \cdot C = (-M(c-s), -M(c+s))$.
\begin{figure}
    \centering
\Description{Illustration of the rotation process of the instance. $A' = R_\theta A$, $B' = R_\theta B$ and $C' = R_\theta C$. In this example the post-rotation coordinate-wise median is $(B'_1, A'_2)$.}    \includegraphics[width=1\linewidth]{"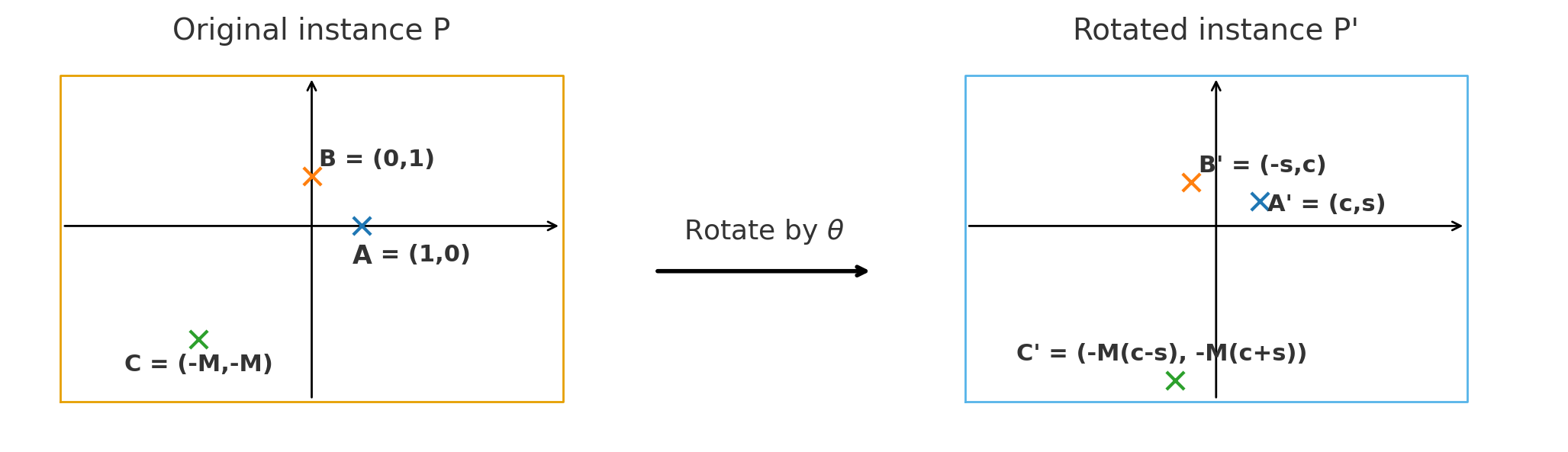"}
    \caption{Illustration of the rotation process of the instance. $A' = R_\theta A$, $B' = R_\theta B$ and $C' = R_\theta C$. In this example the post-rotation coordinate-wise median is $(B'_1, A'_2)$.}
    \label{fig:rotate-2d-illustration}
\end{figure}
Let $m_\theta = (x_\theta, y_\theta)$ denote the coordinate-wise median of the multi-set of rotated points $P'$. Since rotation is an isometry, the social cost of $m_\theta$ w.r.t. $P'$ is the same as the social cost of $m_\theta$ rotated back by $-\theta$ w.r.t. the original multi-set of points $P$.

Since there are $2k+1$ points, the median in each coordinate is the $k+1$-st  order statistic. 


\begin{itemize}
    \item For $\theta \in [0,\frac{\pi}{4})$: $c \ge s$. If $C'_1 \le -s$ then $x_\theta = B'_1 = -s$, and otherwise $x_\theta = C'_1 = -M(c-s)$.

    For the $y$ axis: $y_\theta = s$ always, as the median is $A'_2$.
    
    \item For $\theta \in [\frac{\pi}{4}, \frac{\pi}{2})$: $c \le s$. if $C'_1 \ge c$ then $x_\theta = A'_1 = c$, and otherwise $x_\theta = C'_1 = -M(c-s)$.

    For the $y$ axis: $y_\theta = c$ always, as the median is $B'_2$.
\end{itemize}

Let $W_M := \crl*{\theta \in [0,\frac{\pi}{2}) \mid C'_1 \in (-s, c)}$ be the set of angles for which $x_\theta = C'_1$.
We show that for most angles (all but a $O(\frac{1}{M})$ fraction) it must be that $x_\theta \neq C'_1$:
\begin{clm}\label{clm:W_M-size-ub}
    $|W_M| \le \frac{2}{M}$.
\end{clm}
\begin{proof}
Consider the two possible segments for $\theta$: $(0,\frac{\pi}{4}]$ and $(\frac{\pi}{4},\frac{\pi}{2})$:
\begin{itemize}
    \item For any $\theta \in (0, \frac{\pi}{4}]$:
$C'_1 = -M(c-s) > -s$ if and only if $tan(\theta) > \frac{M}{M+1}$ which implies that:
$(0,\frac{\pi}{4}] \cap W_M = (arctan \frac{M}{M+1}, \frac{\pi}{4})$.

    For any $t$: $|arctan(t) - arctan(1)| \le |t - 1|$ (as $arctan(t)$ is $1$-Lipschitz) and therefore: $\frac{\pi}{4} \le arctan\prn*{\frac{M}{M+1}} + \frac{1}{M}$. It follows that \[\abs*{(0,\frac{\pi}{4}] \cap W_M} \le \frac{1}{M}.\]

\item 
    For any $\theta \in (\frac{\pi}{4}, \frac{\pi}{2})$: $C'_1 = -M(c-s) < c$ if and only if $tan(\theta) < \frac{M+1}{M}$ which implies that: $(\frac{\pi}{4},\frac{\pi}{2}) \cap W_M = [\frac{\pi}{4},arctan(1+\frac{1}{M}))$. Similar to before, $arctan(1+\frac{1}{M}) \le \frac{\pi}{4} + \frac{1}{M}$ and therefore: \[\abs*{(\frac{\pi}{4},\frac{\pi}{2}) \cap W_M} \le \frac{1}{M}.\]
\end{itemize}

Together, we get that $|W_M| \le \frac{1}{M} + \frac{1}{M}\le \frac{2}{M}$.
\end{proof}
We are now ready to lower bound the expected cost of the algorithm.
\begin{clm}\label{clm:expected-cost-of-alg-at-most-4-pi-sqrt-2-times-k}
The expected cost of the algorithm satisfies
\[
\E[ALG] \ge \frac{4}{\pi} k\sqrt{2} - O(\sqrt{k}).
\]
\end{clm}
\begin{proof}
For any $\theta$ the cost of the algorithm is:
\begin{align}\label{eq:ALG_theta}
    ALG_\theta &= k (\norm{m_\theta - A'} + \norm{m_\theta - B'}) + \norm{m_\theta - C'}.
\end{align}
For any $\theta \notin W_M$:
If $\theta \in (0,\frac{\pi}{4}]$, we've shown that $m_\theta = (-s, s)$ and if $\theta \in [\frac{\pi}{4},\frac{\pi}{2})$ then: $m_\theta = (c,c)$.

In the first case: $\norm{A' - m_\theta} = c - s$ and $\norm{B' - m_\theta} = c + s$, and thus $k(\norm{A' - m_\theta} + \norm{B' - m_\theta}) = 2k \ c$.

Similarly in the second case: $\norm{A' - m_\theta} = s+c$ and $\norm{B' - m_\theta} = s - c$ and thus $k(\norm{A' - m_\theta} + \norm{B' - m_\theta}) = 2k \ s$.

We thus get the following bound on $ALG_\theta$:

So for any $\theta \in [0,\frac{\pi}{4}] \ \setminus W_M$: 
\[
    ALG_\theta \ge 2k c, \numberthis \label{eq:ALG_theta-bound-1}
\] and for $\theta \in [\frac{\pi}{4},\frac{\pi}{2}) \ \setminus W_M$: 
\[
    ALG_\theta \ge 2k s.  \numberthis \label{eq:ALG_theta-bound-2}
\]






The expected cost of the algorithm is:
\begin{align*}
     & \E_{\theta \sim U[0,\frac{\pi}{2}]} \brk*{ALG_\theta} = \int_{\theta = 0}^{\frac{\pi}{2}} \frac{2}{\pi} \cdot ALG_{\theta} \ d\theta \\
    & =  \frac{2}{\pi} \prn*{\int_{\theta \in [0, \frac{\pi}{4}] \ \setminus \ W_M} ALG_\theta \ d\theta + \int_{\theta \in [ \frac{\pi}{4},  \frac{\pi}{2}] \ \setminus W_M} ALG_\theta  \ d\theta
    + \int_{\theta \in W_M} ALG_\theta \ d\theta} \\
    & \stackrel{(\star)}{\ge} \frac{2}{\pi} \prn*{\int_{\theta \in [0, \frac{\pi}{4}] \ \setminus \ W_M} 2k c \ d\theta + \int_{\theta \in [ \frac{\pi}{4},  \frac{\pi}{2}] \ \setminus W_M} 2k s \ d\theta
    + \int_{\theta \in W_M} ALG_\theta \ d\theta} \\
    & \ge \frac{2}{\pi} \prn*{ k \sqrt{2} + k \sqrt{2} - O(\sqrt{k})} = \frac{4}{\pi} \cdot k \sqrt{2} - O(\sqrt{k}),
\end{align*}
where:
\begin{enumerate}[label=(\roman*)]
    \item Inequality $(\star)$ follows from \cref{eq:ALG_theta-bound-1,eq:ALG_theta-bound-2}.
    \item The last inequality holds since
    \begin{align*}
        \int_{\theta \in [0, \frac{\pi}{4}] \setminus W_M} 2kc \, d\theta
        & = \int_{0}^{\pi/4} 2kc \, d\theta - \int_{\theta \in W_M} 2kc \, d\theta
        \\
        & \ge k\sqrt{2} - \frac{2k}{M}
        = k\sqrt{2} - 2\sqrt{k},
    \end{align*}
    where the bound uses \cref{clm:W_M-size-ub} and the fact that
    \[
    \int_{0}^{\pi/4} \cos\theta \, d\theta = \sqrt{2}.
    \]
    Similarly,
    \[
    \int_{\theta \in [\pi/4, \pi/2] \setminus W_M} 2ks \, d\theta
    \ge k\sqrt{2} - 2\sqrt{k}.
    \]
\end{enumerate}
\end{proof}

The optimal solution cost is at most the cost of the solution of placing the facility at $A = (1,0)$: $OPT \le \sqrt{2}k + O(M) = \sqrt{2}k + O(\sqrt{k})$ .

Thus, from \cref{clm:expected-cost-of-alg-at-most-4-pi-sqrt-2-times-k}:
\[
\frac{\E[ALG]}{OPT} \ge \frac{\frac{4}{\pi} \ k \sqrt{2} - O(\sqrt{k})}{\sqrt{2} k + O(\sqrt{k})} \stackrel{k \to \infty}{\to} \frac{4}{\pi}.
\]

\end{proof}

\subsection{$\R^d$ approximation ratio analysis}

Next, we show a bound on the approximation ratio of $\rrcwmed$ for arbitrarily large dimension $d$:

\begin{thm}\label{thm:rrcwmed-expected-approximation-ratio-Rd}
Let $d \gg 1$. The following bound holds for the approximation ratio of \cref{alg:rrcwmed} for $\R^d$: \[
\alpha(\rrcwmed) \in \brk*{\sqrt{2} - O\prn*{\frac{1}{\sqrt{d}}}, \sqrt{6 \sqrt{3} - 8}} \approx \brk*{1.41 - O\prn*{\frac{1}{\sqrt{d}}},1.547}.\]
\end{thm}
We show here a proof sketch, deferring the full proof to \cref{sec:proof-of-thm:rrcwmed-expected-approximation-ratio-Rd}.

\begin{proof}[Proof sketch of \cref{thm:rrcwmed-expected-approximation-ratio-Rd}]
For the upper bound, an approach similar to the one we've taken in \cref{lem:R2-ub} for $\R^2$ may lead to a slightly better $1.5$ bound for $\R^3$, but it does not scale well; it leads to a ratio of $\Omega(\sqrt{d})$ for large $d$ values (e.g. see the $\Omega(\sqrt{d})$ analysis for the projection median approximation of \cite{basu2012projection}). The $\sqrt{6 \sqrt{3} - 8}$ upper bound follows directly from the approximation ratio upper bound of $\cwmed$ (see \cref{prop:cwmed-properties}): \citet{gravin2025approximation} show that the approximation ratio of $\cwmed$ is bounded by $\sqrt{6 \sqrt{3} - 8}$ by analyzing a factor revealing optimization problem. Since $\rrcwmed$ is $\cwmed$ with random rotations, its expected cost is upper bounded by the worst-case cost of $\cwmed$.

We now show that a construction analysis which closes most of the remaining gap. While the idea in the construction is very similar to our two dimensional one (two clusters and an outlier construction), the analysis requires ``heavier machinery'' --- we use random matrix theory and high-dimensional probability to show that the expected distance of each cluster point from the random returned point is large.

Let $d \gg 1$ and let $k = \floor{\sqrt{d}}$, $M = \sqrt{k}$. The instance $P$ we observe contains:
$k$ copies of $e_1$, $k$ copies of $e_2$ and a single copy of $-M(e_1 + e_2)$. Our goal is to show that $\E[ALG] \approx 2k$, which implies that the approximation ratio is roughly $\sqrt{2}$ (as $OPT \le k\sqrt{2} + O(\sqrt{k})$).

Let $R \in SO(d)$ be the uniformly sampled rotation matrix, let $a = R e_1$, $b = R e_2$ and $c = R (-M(e_1+e_2)) = -M(a+b)$. Let $m$ be the $\cwmed$ of the rotated multi-set of points $R P$. 
For any $j \in [d]$, let $m_j$ be the median for coordinate $j$. In coordinate $j$ the distinct locations of the rotated points are $a_j, b_j, c_j$. Since \[
    \E[ALG] \ge k \prn*{\E\brk*{\norm{a - m}_2} + \E\brk*{\norm{b - m}_2}}, \numberthis \label{eq:alg-expected-val-sqrt2-lb-of-randcwmed-sketch}
\] we will show that $\E\brk*{\norm{a - m}_2}$ and $\E\brk*{\norm{b-m}_2}$ are approximately $1$, which implies the desired.

For each coordinate $j\in[d]$, define the event $D_j$ that the outlier coordinate $c_j$ lies between $a_j$ and $b_j$
(i.e., $\min\{a_j,b_j\}\le c_j\le \max\{a_j,b_j\}$). On coordinates where $D_j$ fails, the outlier does not affect the median:
$m_j\in\{a_j,b_j\}$. Using an explicit case analysis (formalized as a claim in the full proof), we obtain the inequality
\begin{equation}\label{eq:abm-quad-sketch}
\|a-m\|_2^2+\|b-m\|_2^2
\;\ge\;
\|a-b\|_2^2
-\sum_{j=1}^d \frac{(a_j-b_j)^2}{2}\,\mathbf{1}\{D_j\}.
\end{equation}
Since $a$ and $b$ are orthonormal columns of $R$, we have $\|a-b\|_2^2=\|a\|_2^2+\|b\|_2^2-2\langle a,b\rangle=2$.
Taking expectations in~\eqref{eq:abm-quad-sketch} and using symmetry (that is, $\E\brk*{\norm{a-m}_2^2} = \E\brk*{\norm{b-m}_2^2}$) yields
\begin{equation}\label{eq:Ez2-core-sketch}
\mathbb{E}\|a-m\|_2^2
\;\ge\;
1-\frac{1}{4}\sum_{j=1}^d \mathbb{E}\bigl[(a_j-b_j)^2\,\mathbf{1}\{D_j\}\bigr].
\end{equation}
Thus, the key task is to show that the total ``bad-event mass'' on the RHS is $O(1/k)$.
Define normalized coordinates $
U_j:=\sqrt{\frac d2}(a_j+b_j),\  V_j:=\sqrt{\frac d2}(a_j-b_j)$.
Then $(a_j-b_j)^2=\frac{2}{d}V_j^2$, and it follows from $c_j=-M(a_j+b_j)$ that $D_j$ occurs iff $|U_j|\le t\,|V_j|$ where $t:=\frac{1}{2M+1}$.

We now use the Diaconis--Freedman normal approximation theorem for random orthogonal matrices
(equivalently, the fact that finitely many coordinates of a Haar-uniform $R\in SO(d)$ are close in total variation to i.i.d.\ Gaussians after the $\sqrt d$ scaling) \cite{diaconis1987dozen}.
Concretely, for each fixed $j$, the pair $(U_j,V_j)$ is $O(1/d)$-close in total variation to $(G_1,G_2)$ where $G_1,G_2\stackrel{\text{i.i.d.}}{\sim}N(0,1)$.
We show this implies
\begin{equation}\label{eq:tv-reduction}
\mathbb{E}\!\left[V_j^2\,\mathbf{1}\{|U_j|\le t|V_j|\}\right]
\;=\;
\mathbb{E}\!\left[G_2^2\,\mathbf{1}\{|G_1|\le t|G_2|\}\right] + O(1/\sqrt{d}).
\end{equation}

For $(G_1,G_2)\sim N(0,I_2)$,
we show (by integral calculations) that $\E\brk*{G_1^2 \indic{\abs{G_1} \le t \abs{G_2}}} \le O\prn*{\frac{1}{k}}$.
Combining with~\eqref{eq:tv-reduction} and the identity $(a_j-b_j)^2=\frac{2}{d}V_j^2$, we obtain
\begin{align}
\sum_{j=1}^d \mathbb{E}\bigl[(a_j-b_j)^2\,\mathbf{1}\{D_j\}\bigr]
&=
\sum_{j=1}^d \frac{2}{d}\,\mathbb{E}\!\left[V_j^2\,\mathbf{1}\{|U_j|\le t|V_j|\}\right]\nonumber\\
&\le
2\cdot O(1/M^3)\;+\;O(1/\sqrt{d})\nonumber\\
&=
O(1/k),
\label{eq:bad-mass-sketch}
\end{align}
using $M=\sqrt{k}$ and $d = \Theta( k^2)$.

Plugging~\eqref{eq:bad-mass-sketch} into~\eqref{eq:Ez2-core-sketch} yields
\begin{equation}\label{eq:Ez2}
\mathbb{E}\|a-m\|_2^2\;\ge\;1-O(1/k).
\end{equation}

We wish to deduce\footnote{Note that a random variable $X$ such that $\E\brk*{X^2} \ge 1 - O(\frac{1}{k})$ does not imply that $\E\brk*{X} \ge 1 - O(\frac{1}{k})$ in the general case.} that $\E\brk*{\norm{a-m}_2} \ge 1 - O\prn*{\frac{1}{k}}$.

Let $Z := \norm{a - m}_2$.
We already have the second moment lower bound $\E\brk*{Z^2} \ge 1 - O\prn*{\frac{1}{k}}$, and the goal is to show that $\E\brk*{Z} \ge 1 - O\prn*{\frac{1}{k}}$.
By the definition of the variance of $Z$: 
\begin{equation}\label{eq:sketch-var-Z}
    \prn*{\E[Z]}^2 = \E[Z^2] - Var(Z).
\end{equation}
The idea in the rest of the lemma is to show that $Var(Z) = O\prn*{\frac{1}{k}}$ and thus by we'll get the desired.

Consider the function $f(R):=\|Re_1-\mathrm{CWM}(RP)\|_2$ on $SO(d)$ (so $f(R) = Z$).
We show (in the full proof) that $f$ is $L$-Lipschitz w.r.t.\ the Frobenius metric, with $L=O(M)$.
Then, we apply a concentration inequality for Lipschitz functions of random variables on $SO(d)$ and get that \begin{equation}
    Var(Z) \le  O \prn*{\norm{Z - \E \brk*{Z}}_{\psi_2}^2} \le O\prn*{\norm{f(R) - \E \brk*{f(R)}}_{\psi_2}} \le O\prn*{\frac{1}{k}}.
\end{equation}

Using $\mathbb{E}[Z]^2=\mathbb{E}[Z^2]-\Var(Z)$ together with~\eqref{eq:Ez2} and~\eqref{eq:sketch-var-Z}, we conclude
\begin{equation}\label{eq:EZ}
\mathbb{E}\|a-m\|_2\;\ge\;1-O(1/k).
\end{equation}

Combining~\eqref{eq:alg-expected-val-sqrt2-lb-of-randcwmed-sketch} with~\eqref{eq:EZ} and symmetry,
\[
 \E[ALG] \ge 2k\cdot (1-O(1/k)) = 2k-O(1).
\]
Since $OPT \le \sqrt{2}k + O\prn*{\sqrt{k}}$ we get the desired:
\[
    \frac{\E\brk{ALG}}{OPT} \ge \sqrt{2} - O\prn*{\frac{1}{k}} = \sqrt{2} - O\prn*{\frac{1}{\sqrt{d}}}.
\]
\end{proof}

\cref{thm:rrcwmed-expected-approximation-ratio-Rd} still leaves a gap to close. We finish this section with the following conjecture, which states that our lower bound for $\rrcwmed$ is tight:
\begin{conj}\label{conj:rrcwmed-ub-Rd}
Let $d \in \N$.
The approximation ratio of \cref{alg:rrcwmed} on $\R^d$ is at most $\sqrt{2}$.
\end{conj}

\section{Improved Learning Augmented Facility Location Mechanism Design Via Randomness}\label{sec:predictions-learning-aug}

In this section we show how our $\rrcwmed$ mechanism yields better results in the learning augmented setting for the facility location mechanism design problem in $\R^2$. We study two prediction models: the output prediction, and the $\mac$ input predictions models.

\paragraph{Output prediction setting}
For this setting, \cite{agrawal2022learning} give the $\cmp$ mechanism which adds $cn$ copies of the predicted point $\hat{g}$ to the dataset (for $c \in (0,1)$) and returns the $\cwmed$ of the new multiset of points. By replacing the $\cwmed$ with $\rrcwmed$, we maintain strategyproofness and obtain a better consistency-robustness trade-off in expectation for many values of $c$.
We name the resulting mechanism $\rrcmp$ (\cref{alg:rrcmp}).

\begin{algorithm}[H]
\caption{$\rrcmp$ mechanism}
\label{alg:rrcmp}
\begin{algorithmic}[1]
\Require Multiset $P=\{p_1,\ldots,p_n\}\subset\mathbb{R}^d$, prediction $\hat g\in\mathbb{R}^d$, confidence parameter $c\in[0,1)$.
\State Create a multiset $\hat{P}$ that contains $\floor{cn}$ copies of $\hat{g}$.
\State Return $\rrcwmed(P \cup \hat{P})$.
\end{algorithmic}
\end{algorithm}

In the following theorem, we show that the $\rrcmp$ has better consistency-robustness trade-off by inheriting the consistency of $\cmp$ mechanism with improved robustness, achieved via the random rotations.

\begin{thm}\label{thm:rrcmp-approximation}
    The $\rrcmp$ mechanism with parameter $c \in (0,1)$ is strategyproof and achieves an approximation ratio of at most:
    \[
        \alpha(\rrcmp) \le \min\crl*{\frac{\sqrt{2c^2 + 2}}{1+c} + \eta, \frac{4}{\pi} \prn*{1 + c \eta}, \frac{\sqrt{2c^2 + 2}}{1 - c}, \frac{4}{\pi} \cdot \frac{1 + c}{1-c}}.
    \]
    Thus, $\rrcmp$ is $\min\crl*{\frac{\sqrt{2c^2 + 2}}{1+c}, \frac{4}{\pi}}$ consistent, $\frac{1}{1-c} \min\crl*{\sqrt{2c^2 + 2}, \frac{4}{\pi} \prn*{1 + c}}$ robust and has a smooth degradation between consistency and robustness as a function of the prediction error $\eta$.
\end{thm}

A calculation shows that \cref{thm:rrcmp-approximation} implies a strict improvement in the expected consistency-robustness trade-off over $\cmp$: for $c \in (0, 0.118], \eta \ge 0$ or $c \in [0.118, 0.785), \eta > \frac{4(1+c)-\pi\sqrt{2c^2+2}}{(\pi-4c)(1+c)}$.
\citet{agrawal2022learning} shows a bound on any consistency-robustness trade-off of deterministic anonymous strategyproof mechanisms, showing that the trade-off of $\cmp$ is optimal for \emph{any} value $c \in (0,1)$. Thus, \cref{thm:rrcmp-approximation} implies breaking the deterministic barrier in the learning-augmented setting as well.

Before showing the proof of \cref{thm:rrcmp-approximation}, we show a robust statistics result about the robustness of the $\cwmed$  to a $c$-``insertion corruption'' --- adding a $c$ fraction of arbitrary points to the dataset. In fact, we give robustness results for a generalization of the median (and of $\cwmed$ and $\geomed$) to a general metric space: the $\OneMed$.
\begin{definition}[$\OneMed$]
    Let $(V,d)$ be a metric space.
    The $\OneMed$ of a multiset of points $P \subset V$ is the minimizer of \ $C(P,t) := \sum_{p \in P} d(p,t)$. That is: \ $\OneMed(P) := \argmin_{t \in V} C(P,t)$.
\end{definition}

We give a bound on the approximation ratio of the solution computed on the corrupted dataset:
\begin{lem}\label{lem:geometric-median-approx-robustness}
    Let $c \in (0,1)$ and $(V,d)$ be a metric space. Let $P, \hat{P} \subset V$ be two multisets of points such that $\abs{\hat{P}} = c \abs{P}$. Then: 
    \[
    C\prn*{P,med\prn*{P \cup \hat{P}}} \le \prn*{1 + \frac{2c}{1-c}} C\prn*{P,med\prn*{P}},
    \]
    where for any $X \subset V$: $C(X,t) = \sum_{x_i \in X} d(x_i,t)$, and $med(X) := \argmin_{t \in V} C(X,t)$.
\end{lem}

We are now ready to prove \cref{thm:rrcmp-approximation}.
\begin{proof}
Strategyproof is immediate from the fact that the prediction is independent from the agent reports, and from the strategyproofness of $\rrcwmed$.
Let $ALG_\theta$ denote the cost of $\rrcmp$ given any sampled $\theta$, and let $OPT$ be the cost of the optimal solution.

The approximation ratio of $\rrcmp$ is upper bounded by that of $\cmp$. Specifically, for any instance $I$ and any $\theta \in [0, \frac{\pi}{4})$ sampled by $\rrcmp$, the cost incurred by $\rrcmp$ on instance $I$ equals the cost incurred by $\cmp$ on the rotated instance $I_\theta$, obtained from $I$ by rotating all points by $\theta$. Consequently, $\rrcmp$ directly inherits the approximation guarantees of $\cmp$ established by \citet{agrawal2022learning}:
\[\E[ALG_\theta] \le \min\crl*{\frac{\sqrt{2c^2 + 2}}{1+c} + \eta, \frac{\sqrt{2c^2+2}}{1-c}} OPT. \]

W.l.o.g., assume that $\hat{g} = 0$. Otherwise, we may translate the instance by $t = -\hat{g}$. The rotated translated instance is equivalent to the original rotated instance translated by $R_{\theta} t$, where $R_{\theta}$ denotes the rotation matrix. Consequently, the resulting $\rcwmed$ is equal to the original one, translated by $R_{\theta} t$. After rotating back by $R_{-\theta}$, we obtain the same solution translated by $t$. Therefore, all pairwise distances remain unchanged, and the transformed instance preserves the same approximation ratio.

Let $g$ be the geometric median. Then $OPT = \sum_{i=1}^n \norm{p_i - g}_2$.
Let $m = \floor{cn} \in \N$.

Let $\hat{h}_\theta$ be the coordinate-wise median of the rotated dataset $R_\theta(P \cup \hat{P})$, and let $h'_\theta$ be the coordinate-wise median of the rotated dataset without the copies of the prediction: $R_\theta P$.
Let $h$ be the location returned by $\rrcmp$.
We start in a similar way to the proof of \cref{lem:R2-ub}:
\begin{align*}
    ALG_\theta & = \sum_{i\in[n]} \norm{p_i - h}_2 = \sum_{i\in[n]} \norm{R_\theta p_i - R_\theta h}_2
     \le \sum_{i\in[n]} \norm{R_\theta p_i - \hat{h}_\theta}_1, \numberthis \label{eq:ALG-theta-tmp-bound-learning-aug}
\end{align*}
where the second equality holds as $\norm{\cdot}_2$ is invariant to rotation and the inequality is due to the fact that $\norm{x}_2 \le \norm{x}_1$ for any $x \in \R^2$.

First, we show the consistency improvement. By \cref{eq:ALG-theta-tmp-bound-learning-aug}:

\begin{align*}
    ALG_\theta &\le \sum_{i\in[n]} \norm{R_\theta p_i - \hat{h}_\theta}_1 + m \norm{R_\theta \hat{g} - \hat{h}_\theta}_1 - m \norm{R_\theta \hat{g} -\hat{h}_\theta}_1 \\
    & \le \sum_{i\in[n]} \norm{R_\theta p_i - R_\theta g}_1 + m \norm{R_\theta \hat{g} - R_\theta g}_1 - m \norm{R_\theta \hat{g} - \hat{h}_\theta}_1 \\
    & = \sum_{i\in[n]} \norm{R_\theta p_i - R_\theta g}_1 + m \prn*{\norm{R_\theta g}_1 - \norm{\hat{h}_\theta}_1}, \numberthis \label{eq:alg-theta-rr-cmp-tmp}
\end{align*}

where the inequality is due to the fact that $\hat{h}_\theta$ is the $\scost_1$ minimizer solution of $R_\theta \prn*{P \cup \hat{P}}$, as its $\cwmed$ (due to \cref{prop:cwmed-properties}), and thus $R_\theta \ g$ has a larger $\scost_1$ cost.
The last equality holds as $\hat{g} = 0$.

From \cref{eq:alg-theta-rr-cmp-tmp} and the fact that $m = \floor{cn}$:
\begin{equation}\label{eq:rrcmp-expected-ub-tmp}
    \E[ALG_\theta] \le \E\brk*{\sum_{i\in[n]} \norm{R_\theta p_i - R_\theta g}_1} + cn \E\brk*{\norm{R_\theta g}_1 - \norm{\hat{h}_\theta}_1}.
\end{equation}

$\norm{R_\theta g}_1 - \norm{\hat{h}_\theta}_1 \le \norm{R_\theta g}_1 = \norm{R_\theta g - R_\theta \hat{g}}_1$, and thus, by plugging \cref{lem:expected-abs-value-of-each-rotated-component} in \cref{eq:rrcmp-expected-ub-tmp}:
\begin{align*}
    \E[ALG_\theta] &\le \frac{4}{\pi} OPT + c n \frac{4}{\pi} \E \brk*{\norm{R_\theta g - R_\theta \hat{g}}_2} = \frac{4}{\pi} OPT + c n \frac{4}{\pi} \norm{g - \hat{g}}_2 \\
    & = \frac{4}{\pi} \prn*{1 + c \eta} OPT,
\end{align*}
where the one before last equality holds by the fact rotation is an isometry, and the last equality is due to the definition of $\eta$.

For the robustness upper bound, by \cref{eq:ALG-theta-tmp-bound-learning-aug} and \cref{lem:geometric-median-approx-robustness}:
\begin{align*}
    ALG_\theta & \le \sum_{i\in[n]} \norm{R_\theta p_i - \hat{h}_\theta}_1 = \scost_1(R_\theta P, \hat{h}_\theta) \le \prn*{1 + \frac{2c}{1-c}}\scost_1(R_\theta P, h'_\theta) \\
    & \le \prn*{1 + \frac{2c}{1-c}}\scost_1(R_\theta P, R_\theta g) = \prn*{1 + \frac{2c}{1-c}} \sum_{p \in P} \norm{R_\theta (p - g)}_1,
\end{align*}
where the last inequality is due to the fact that $h'_\theta$ is the $\scost_1$ minimizer solution of dataset $R_\theta P$.
By taking expectation and using \cref{lem:expected-abs-value-of-each-rotated-component} we get:
\[
    \E\brk*{ALG_\theta} \le \prn*{1 + \frac{2c}{1-c}} \E\brk*{\sum_{p \in P} \norm{R_\theta (p - g)}_1} = \prn*{1 + \frac{2c}{1-c}} \frac{4}{\pi} \sum_{p \in P} \norm{p - g}_2 = \frac{4}{\pi} \frac{1+c}{1-c} OPT.
\]

\end{proof}

\paragraph{$\mac$ input predictions settings} In this setting \citet{barak2024mac} give a strategyproof mechanism that chooses the best option between using only the predictions (achieving $1 + \frac{4\delta}{1-2\delta}$) and simply using the $\cwmed$ mechanism (achieving $\sqrt{2}$), which implies an approximation ratio of at most $\min\crl*{1 + \frac{4\delta}{1-2\delta}, \sqrt{2}}$. By replacing the $\cwmed$ with our $\rrcwmed$ we immediately obtain a better result as a direct corollary of \cref{thm:alg-R2}: $\min\crl*{1 + \frac{4\delta}{1-2\delta}, \frac{4}{\pi}}$.

\section{The Limitation Of Generalized Random Dictator Mechanisms}\label{sec:lower-bounds}
In this section we show the limitations of generalized random dictator mechanisms.

\begin{definition}[Generalized Random Dictator (GRD) Mechanism]\label{def:GRD}
    We say a random mechanism $M$ is GRD (Generalized Random Dictator) if it may only output one of the points reported by the agents. That is, for any $P \subseteq \R^d$: $M(P) \in P$ w.p. $1$.
\end{definition}

Note that the above definition holds both for truthful and non-truthful GRD mechanisms.

\begin{thm}\label{thm:lb-rand-dict-R2}
    Any GRD mechanism for the utilitarian facility location problem in $\R^2$ has an approximation ratio of at least $\frac{4}{\pi} \approx 1.27$.
\end{thm}

\begin{proof}

Let $n \gg 1$.
Consider the instance $P = \crl*{p_0, \ldots ,p_{n-1}}$ where $p_k = \prn*{cos(\frac{2\pi k}{n}), sin(\frac{2\pi k}{n})}$ for $k \in \crl*{0,\ldots,n-1}$.
The point $y=0$ has cost $\scost(P,y) = \sum_{i=0}^{n-1} 1 = n$ and thus the optimal solution has cost at most $n$: $OPT \le n$.
Let $A$ be a GRD mechanism, and let $ALG$ be its cost.
We show that for any dictator choice $p_k \in P$ of $A$, the social cost of this solution is approximately $\frac{4}{\pi} n$.

Assume w.l.o.g. that $p_k = p_0$ (due to symmetry the cost is the same). The distance between $p_0$ and $p_j$ separated by the angle $\frac{2\pi j}{n}$ is the chord length:
\[
    \norm{p_0 - p_j}_2 = 2 sin\prn*{\frac{\pi j}{n}}.
\]

Therefore $ALG = 2 \sum_{j=1}^{n-1} sin\prn*{\frac{\pi j}{n}} =2 \ cot(\frac{\pi}{2n})$ where the last equality follows from the identity:
\begin{clm}\label{clm:trigo-identity}
    \[
        \sum_{k=1}^{n-1}\sin \prn*{\frac{\pi k}{n}} = \cot \prn*{\frac{\pi}{2n}}.
    \]
\end{clm}
The proof of \cref{clm:trigo-identity} follows from standard trigonometric identities and the sum of a telescoping series, and we provide it in \cref{sec:missing-proofs-rand-dict-lb}.

By the Taylor expansion of $cot$, the cost of $A$ is:
\[
ALG = 2 \ cot(\frac{\pi}{2 n}) = 2 \prn*{\frac{2n }{\pi} - O\prn*{\frac{1}{n}}} = \frac{4n}{\pi} - O\prn*{\frac{1}{n}}.\]

Since the above holds for any $p_k$, the above equation yields the desired by taking $n \to \infty$.

\end{proof}

We now show a stronger lower bound $d \gg 1$.

\begin{thm}\label{thm:lb-rand-dict-Rd}
Let $d > 2$.  Any GRD mechanism for the utilitarian facility location problem in $\R^d$ has an approximation ratio of at least $\sqrt{2}\prn*{1 - O\prn*{\frac{1}{d}}}$.
\end{thm}
\begin{proof}

The idea is similar to the one of \cref{thm:lb-rand-dict-R2}, but we use the probabilistic method \cite{alon2016probabilistic} rather than exhibiting an explicit worst-case instance: instead of choosing evenly spaced points, we sample i.i.d.\ points from the unit sphere. We will lower bound the algorithm's cost with high probability and upper bound the optimum deterministically; by the probabilistic method there exists a fixed multi-set achieving the same bounds.

Let $n \gg 1$. Let us observe the instance of points $P = \crl*{Y_1, \ldots,Y_{n+1}} \subset \R^d$ where $Y_1, \ldots,Y_{n+1}$ are i.i.d.\ random variables such that for each $i \in [n+1]$: $Y_i \sim U\crl*{S^{d-1}}$ is chosen randomly from the unit sphere.
Let $A$ be a GRD mechanism with cost $ALG$.

Fix $j \in [n+1]$. Without loss of generality assume $Y_j = e_1 = (1,0,\ldots,0) \in \R^d$: indeed, condition on $Y_j$ and apply an orthogonal map $Q$ with $QY_j=e_1$; the joint law of $\crl*{QY_i}_{i\ne j}$ is the same as that of $\crl*{Y_i}_{i\ne j}$ and all pairwise distances are preserved. Consider $ALG_j$, the sum of distances of the points of $P$ from point $Y_j = e_1$. This is the cost of mechanism $A$ in case it returns point $Y_j$. Then
\[
\E[ALG_j] = \E_P [ALG_j] = \E_P \brk*{\sum_{Y_i \in P} \norm{Y_i - e_1}_2} = n \ \E_{X \sim U\prn*{S^{d-1}}}\brk*{\norm{X - e_1}_2}. \numberthis \label{eq:expected-sum-of-distances-from-e1}
\]

The following lemma bounds $\E\brk*{\norm{X-e_1}_2}$ as a function of the even moments of $X_1$ where $X = (X_1,\ldots,X_d) \sim U\prn*{S^{d-1}}$:

\begin{lem}[Expected distance between random sphere vectors]\label{lem:random-sphere-vector-distance-from-e1}
Let $X=(X_1,\dots,X_d) \in \R^d$ s.t. $X \sim U\prn*{S^{d-1}}$ is a uniform random variable on the unit sphere. Then:
\[
\E\,\norm{X-e_1}_2
=\sqrt{2}\!\left(1-\frac18\,\mathbb{E}[X_1^2]
-\frac{5}{128}\,\mathbb{E}[X_1^4]
-\frac{21}{1024}\,\mathbb{E}[X_1^6]-\cdots\right).
\]
\end{lem}
We defer the proof of \cref{lem:random-sphere-vector-distance-from-e1} to \cref{sec:proof-of-expected-rand-sphere-vectors-lemma}.

The next lemma states a known fact about the even moments of $X_1$ (see the proof of Proposition 2.5 in \cite{meckes2019random}):

\begin{lem}[Even moments on the sphere]\label{lem:even-moments-sphere}
For every integer $k\ge 0$,
\[
\mathbb{E}[X_1^{2k}]
=\prod_{j=1}^{k}\frac{2j-1}{d+2j-2}
=\frac{1\cdot3\cdot5\cdots(2k-1)}{d(d+2)\cdots(d+2k-2)}.
\]
\end{lem}

Plugging \cref{lem:even-moments-sphere} into \cref{lem:random-sphere-vector-distance-from-e1} gives:
\[
\mathbb{E}\,\norm{X-e_1}_2
=\sqrt{2}\left(
1-\frac{1}{8d}
-\frac{15}{128\,d(d+2)}
-\frac{315}{1024\,d(d+2)(d+4)}-\cdots
\right).
\]
In particular,
\[
\E \brk*{\norm{X-e_1}_2} = \sqrt{2}\prn*{1-\frac{1}{8d}}+O\!\prn*{\frac{1}{d^2}}.
\]

Denote $\mu_d := \E \brk*{\norm{X-e_1}_2} = \sqrt{2}\prn*{1-\frac{1}{8d}}+O\!\prn*{\frac{1}{d^2}}$. By \cref{eq:expected-sum-of-distances-from-e1} we have
\[
    \E[ALG_j] \;=\; n\,\mu_d.
\]

$ALG_j$ is a sum of independent random variables that lie in $[0,2]$. By Hoeffding's inequality \cite{hoeffding1963probability}, for any $t>0$,
\[
\Pr\prn*{ALG_j \le n\mu_d - t} \;\le\; \exp\prn*{-\frac{t^2}{2n}}.
\]
By a union bound over $j\in[n+1]$,
\[
\Pr\prn*{\min_{j\in[n+1]} ALG_j \le n\mu_d - t} \;\le\; (n+1)\,\exp\prn*{-\frac{t^2}{2n}}.
\]
Taking $t=\sqrt{2n}\log\prn*{n+1}$ gives, with probability at least $1-\tfrac{1}{n+1}$,
\[
\min_{j\in[n+1]} ALG_j \;\ge\; n\mu_d - t.
\]

Since a GRD mechanism returns one of the points $Y_j$, $\E_R\brk*{ALG}$ (expectation over the internal randomness $R$ of $A$) is a convex combination of the $ALG_j$, and hence
\[
\E_R\brk*{ALG} \;\ge\; \min_{j\in[n+1]} ALG_j.
\]
Therefore, w.p. at least $1-\tfrac{1}{n+1}$ over the draw of $P$, $\E_R\brk*{ALG} \ge n\mu_d - \sqrt{2n}\log\prn*{n+1}$.
By the probabilistic method, there exists a deterministic multi-set $P$ for which the above inequality holds.

On the other hand, in any such instance $\mathrm{OPT} \le \abs{P}=n+1$ (placing the facility at the origin yields cost $n+1$). Hence for that $P$,
\[
\frac{\E_R\brk*{ALG}}{\mathrm{OPT}}
\ge
\frac{n\mu_d - \sqrt{2n}\log\prn*{n+1}}{n+1}
=
\mu_d \frac{n}{n+1}  - O \prn*{\tfrac{\log n}{\sqrt{n}}}.
\]
Letting $n\to\infty$ (with $d$ fixed) and recalling $\mu_d=\sqrt{2}\prn*{1-\frac{1}{8d}}+O \prn*{\frac{1}{d^2}}$ proves
\[
\frac{\E_R\brk*{ALG}}{\mathrm{OPT}}
\;\ge\;
\sqrt{2}\prn*{1 - O\prn*{\frac{1}{d}}}.
\]
This completes the proof.
\end{proof}

\section{Acknowledgments}\label{sec:ack}

We are grateful to Inbal Talgam Cohen for useful discussions and support.

We would like to thank the anonymous reviewers for their helpful feedback that improved the paper.

This work was supported by the European Research Council (ERC) under the European Union’s Horizon 2020 research and innovation program (grant agreement No. 101077862, project ALGOCONTRACT), and by the Israel Science Foundation (grant No. 3331/24).

\newpage
\appendix
\section*{Organization of Appendices}


In \cref{sec:missing-analysis-rand-mech} we give missing proofs from \cref{sec:rand-mech}.

In \cref{sec:appendix-learning-aug} we give missing proofs from \cref{sec:predictions-learning-aug}.

In \cref{sec:missing-proofs-rand-dict-lb} we give missing proofs from \cref{sec:lower-bounds}.

In \cref{sec:proj-median-connection} we explain the connection between the approximation ratio of $\rrcwmed$ and the one of $\projmed$.

\section{Missing analysis from \cref{sec:rand-mech}}\label{sec:missing-analysis-rand-mech}

\subsection{Proof of \cref{thm:rrcwmed-expected-approximation-ratio-Rd}}\label{sec:proof-of-thm:rrcwmed-expected-approximation-ratio-Rd}

\begin{proof}
For the upper bound, an approach similar to the one we've taken in \cref{lem:R2-ub} for $\R^2$ may lead to a slightly better $1.5$ bound for $\R^3$, but it does not scale well: it leads to a ratio of $\Omega(\sqrt{d})$ for large $d$ values (e.g. see the $\Omega(\sqrt{d})$ analysis for the projection median approximation of \cite{basu2012projection}).

The $\sqrt{6 \sqrt{3} - 8}$ upper bound follows directly from the approximation ratio upper bound of $\cwmed$ (see \cref{prop:cwmed-properties}): \citet{gravin2025approximation} show that the approximation ratio of $\cwmed$ is bounded by $\sqrt{6 \sqrt{3} - 8}$ by analyzing a factor revealing optimization problem. Since $\rrcwmed$ is $\cwmed$ with random rotations, its expected cost is upper bounded by the worst-case cost of $\cwmed$.

We now show that a construction analysis which closes most of the remaining gap. While the idea in the construction is very similar to our two dimensional one (two clusters and an outlier construction), the analysis requires ``heavier machinery'' --- we use a result of \citet{diaconis1987dozen} in random matrix theory which shows that most of the coordinates of the rotated points roughly (up to small total variation distance) follow standard i.i.d. Gaussian distributions for large $d$ values. We use this to get a lower bound on the expected squared distance of a cluster point from the resulting median. Then we use concentration of measure results in high-dimensional probability to show that the variance of the distance is low, to deduce that the expected distance is low as well.

Let $d \gg 1$ and let $k = \floor{\sqrt{d}}$, $M = \sqrt{k}$. The instance $P$ we observe contains:
$k$ copies of $e_1$, $k$ copies of $e_2$ and a single copy of $-M(e_1 + e_2)$. 

Like before, $OPT$ is upper bounded by $\sqrt{2} k + O(\sqrt{k})$. Our goal is now to show that $\E[ALG] \approx 2k$, which implies that the approximation ratio is roughly $\sqrt{2}$.

Let $R \in SO(d)$ be the uniformly sampled rotation matrix, let $a = R e_1$, $b = R e_2$ and $c = R (-M(e_1+e_2)) = -M(a+b)$. Let $m$ be the $\cwmed$ of the rotated multi-set of points $R P$. 
For any $j \in [d]$: Let $m_j$ be the median in coordinate $j$. In coordinate $j$ the distinct locations of the rotated points are $a_j, b_j, c_j$.

Since \[
    \E[ALG] \ge k \prn*{\E\brk*{\norm{a - m}_2} + \E\brk*{\norm{b - m}_2}}, \numberthis \label{eq:alg-expected-val-sqrt2-lb-of-randcwmed}
\] we will show that $\E\brk*{\norm{a - m}_2}$ and $\E\brk*{\norm{b-m}_2}$ are approximately $1$, which implies the desired (as $OPT = k\sqrt{2} + O(\sqrt{k})$).

$\E\brk*{\norm{a-m}_2} = \E\brk*{\norm{b-m}_2}$ due to symmetry, and thus it is enough to show that $\E\brk*{\norm{a-m}_2} \approx 1$. More accurately, we show that $\E\brk*{\norm{a-m}_2} \ge 1 - O\prn*{\frac{1}{\sqrt{k}}}$.

Let $D_j$ denote the ``bad'' (rare) event where $c_j$ lies between $a_j$ and $b_j$.

We first give the two following technical claims:

\begin{clm}\label{clm:Dj-equiv-condition}
    $D_j$ occurs iff \begin{equation}
    \abs{a_j + b_j} \le \frac{\abs{a_j - b_j}}{2M + 1}.
\end{equation}
\end{clm}

\begin{clm}\label{clm:tech-clm-event-Dj}
    \[
    \prn*{a_j - m_j}^2 + \prn*{b_j - m_j}^2 = (a_j - b_j)^2 + 2\prn*{\prn*{m_j - \frac{a_j + b_j}{2}}^2 - \frac{\prn*{a_j-b_j}^2}{4}} \indic{D_j}.
    \]
\end{clm}

The proofs of \cref{clm:Dj-equiv-condition} and \cref{clm:tech-clm-event-Dj} are quite straightforward, and we defer them to \cref{sec:proof-of-clm-d-j-equiv-cond}, \cref{sec:proof-clm-tech-clm-event-Dj}.

We use \cref{clm:tech-clm-event-Dj} and get:
\begin{align*}
    \norm{a-m}_2^2 + \norm{b-m}_2^2 & = \sum_{j=1}^d \prn*{a_j - m_j}^2 + \prn*{b_j - m_j}^2 \\
    & =\sum_{j \in [d]} (a_j - b_j)^2 \\
    & \ \ \ \ + 2 \sum_{j \in [d]} \prn*{\prn*{m_j - \frac{a_j + b_j}{2}}^2 - \frac{\prn*{a_j-b_j}^2}{4}} \indic{D_j} \\
    & \ge 2 - \sum_{j \in [d]} \prn*{ \frac{\prn*{a_j-b_j}^2}{2}} \indic{D_j},
\end{align*}

where the inequality is due to the fact that $\prn*{m_j - \frac{a_j + b_j}{2}}^2 \ge 0$ and the fact that $\norm{a-b}_2^2 = \norm{a}_2^2 - 2 \langle a,b\rangle + \norm{b}_2^2 = 2$ as $\norm{a}_2 = \norm{b}_2 = 1$ and $a,b$ are orthogonal.

Since $\E\brk*{\norm{a-m}_2^2} = \E\brk*{\norm{b-m}_2^2}$ (due to symmetry), the above inequality yields:
\begin{equation}\label{eq:bound-norm-of-a-minus-m}
    \E\brk*{\norm{a-m}_2^2} \ge 1 - \frac{1}{2} \E \brk*{\sum_{j \in [d]} (a_j - b_j)^2 \indic{D_j}} = 1 - \frac{1}{2} \sum_{j \in [d]} \E\brk*{(a_j - b_j)^2 \indic{D_j}}.
\end{equation}

In the next lemma, we show that the last term in \cref{eq:bound-norm-of-a-minus-m} is negligible:

\begin{lem}\label{lem:negligible-a_j-minus-b_j}
For any $j \in [d]$: 
\[
    \sum_{j \in [d]} \E\brk*{(a_j - b_j)^2 \indic{D_j}} = O\prn*{\frac{1}{k}}.
\]
\end{lem}
\begin{proof}
Let $U_j = \frac{\sqrt{d}}{2} \prn*{a_j + b_j}$ and $V_j = \frac{\sqrt{d}}{2} \prn*{a_j - b_j}$. Then \begin{equation}\label{eq:tmp-a_j-minus-b_j}
    \prn*{a_j - b_j}^2 = \frac{2}{d} V_j^2.
\end{equation}

Event $D_j$ is equivalent (via \cref{clm:Dj-equiv-condition}) to:
\begin{equation}\label{eq:tmp-D_j-equiv}
    D_j = \crl*{ \abs{U_j} \le \frac{\abs{V_j}}{2M+1}}.
\end{equation}

By exchangeability of the rows of $R$, every $j$ has the same marginal row law and thus by \cref{eq:tmp-a_j-minus-b_j}, \cref{eq:tmp-D_j-equiv}:
\[ 
    \sum_{j \in [d]} \E\brk*{(a_j - b_j)^2 \indic{D_j}} = 2 \E\brk*{V_1^2 \indic{\abs{U_1} \le \frac{\abs{V_1}}{2M+1}}}.
\]
To get a bound on the last term, we first argue that $U_1$, $V_1$ follow standard i.i.d Gaussian distributions up to negligible total variation distance.
\begin{definition}($\mathrm{TV}$ distance)\label{def:tv-dist}
    For distributions $P,Q$ on a common space $\Omega$, the total variation (or $\mathrm{TV}$) distance between P and Q is:
    \[
    d_{\mathrm{TV}}(P,Q)\ :=\ \sup_{A \in \Omega}\,|P(A)-Q(A)|.
    \]
\end{definition}
Since $R \in SO(d)$ is the uniform rotation matrix we sample, the rows of $R$ are vectors uniformly drawn from the unit sphere $S^{d-1}$ (see \cite{meckes2019random} for more information on random matrix theory). 
\citet{diaconis1987dozen} show that for any constant $k$, the first $k$ coordinates of a vector drawn randomly (uniformly) from the sphere with radius $\sqrt{d}$ behave as independent standard Gaussian random variables ($N(0,1)$) up to a total variation distance of $O(\frac{1}{d})$ (see the result of \citet{diaconis1987dozen} stated as Theorem 2.8 in \cite{meckes2019random}).
Thus, for any row $R_j$ in $R$, it must be that $\sqrt{d} \ R_{j1}, \sqrt{d} \ R_{j2} \stackrel{i.i.d}{\sim} N(0,1)$, up to the $O(\frac{1}{d})$ TV distance. Since $a_j = R_{j1}$, $b_j = R_{j2}$ for any $j \in [d]$, and since $(V_j, U_j) = T(\sqrt{d} \ a_j,\sqrt{d} \ b_j)$ where $T: \R^2 \to \R^2$, $T(x,y) = \frac{1}{\sqrt{2}}(x+y , x-y)$ is an orthogonal (rotation) map, then the same applies for $V_j, U_j$. That is, $V_j, U_j \stackrel{i.i.d}{\sim} N(0,1)$  up to $O(\frac{1}{d})$ TV distance.

The next claim shows that to bound the $\E\brk*{V_1^2 \indic{\abs{U_1} \le \frac{\abs{V_1}}{2M+1}}}$ we may analyze ``pure'' Gaussian variables.

\begin{clm}\label{clm:convergence-to-gaussians}
$V_j, U_j \stackrel{i.i.d}{\sim} N(0,1)$  up to $O\prn*{\frac{1}{d}}$ TV distance (i.e. $d_{TV} \prn*{\mathcal{L}(V_1, U_1), (G_1,G_2)} = O\prn*{\frac{1}{d}}$).
Let $G_1, G_2 \stackrel{i.i.d}{\sim} N(0,1)$, and $t = \frac{1}{2M+1}$. Then:
\[
    \E\brk*{V_1^2 \indic{\abs{U_1} \le \frac{\abs{V_1}}{2M+1}}} = \E\brk*{G_1^2 \indic{\abs{G_1} \le t \abs{G_2}}} + O\prn*{\frac{1}{\sqrt{d}}}.
\]
\end{clm}
We defer the proof of \cref{clm:convergence-to-gaussians} to \cref{sec:proof-of-clm-convergence-to-gauss}.

By \cref{clm:convergence-to-gaussians}, it is enough to show that:
\[
     \E\brk*{G_1^2 \indic{\abs{G_1} \le t \abs{G_2}}} \le O\prn*{\frac{1}{k}},
\]
where $G_1, G_2 \stackrel{i.i.d}{\sim} N(0,1)$, and $t = \frac{1}{2M+1}$.

Indeed, let $G_1, G_2 \stackrel{i.i.d}{\sim} N(0,1)$. $(G_1,G_2)$ is an uncorrelated and zero centered bivariate normally distributed vector, and thus has density of $f(x,y): \R^2 \to \R$ of $f(x,y) = \frac{1}{2\pi} e^{-(x^2+y^2)/2}$. We switch to polar coordinates and get $x = Rcos\Theta$, $y = Rsin(\Theta)$, $R \ge 0, \Theta \in [0,2\pi)$.
By the change-of-variables theorem, for $r \ge 0, \theta \in [0,2\pi)$ we get:
 \[
 f(r,\theta) = \frac{1}{2\pi} e^{-r^2} r = r \ e^{-\frac{r^2}{2}} \ \indic{r \ge 0} \cdot \frac{1}{2\pi} \ \indic{\theta \in [0,2\pi)]} = f_1(r) \cdot f_2(\theta),
 \]
where $f_1(r) = r \ e^{-\frac{r^2}{2}} \ \indic{r \ge 0}$ and $ f_2(\theta) = \frac{1}{2\pi} \ \indic{\theta \in [0,2\pi)]}$, which implies $(G_1,G_2) = (R cos\Theta, Rsin(\Theta)$ where $R^2 \sim \chi_2^2$ and $\Theta \sim \mathrm{Unif}\crl*{[0,2\pi)}$ are independent random variables.
Let $\alpha = arctan(\frac{1}{t})$.

\begin{align*}
    \E\brk*{G_1^2 \indic{\abs{G_1} \le t \abs{G_2}}} &= \E\brk*{R^2 cos^2(\Theta) \cdot \indic{\abs{cos(\Theta)}} \le t \abs{sin(\Theta)}} \\
    & = \E[R^2] \cdot \E\brk*{ cos^2(\Theta) \cdot \indic{\abs{cos(\Theta)}} \le t \abs{sin(\Theta)}} \\
    & = 2 \cdot \frac{1}{2\pi} \int_{0}^{2\pi} cos^2(\theta) \indic{\abs{tan(\theta)} \ge \frac{1}{t}} d\theta \\
    & = \frac{2}{\pi} \int_{0}^{\pi} cos^2(\theta) \indic{\abs{tan(\theta)} \ge \frac{1}{t}} d\theta \stackrel{(\star)}{=} \frac{4}{\pi} \int_{0}^{\frac{\pi}{2}} cos^2(\theta) \indic{\abs{tan(\theta)} \ge \frac{1}{t}} d\theta \\
    & \stackrel{(\star \star)}{=} \frac{4}{\pi} \int_{\alpha}^{\frac{\pi}{2}} cos^2(\theta) d\theta = \frac{2}{\pi} \int_{\alpha}^{\frac{\pi}{2}} 1 + cos(2\theta) d\theta \\
    & = \frac{2}{\pi} \prn*{\frac{\pi}{2} - \alpha - \frac{sin(2\alpha)}{2}} = O\prn*{\frac{1}{M^3}} = O\prn*{\frac{1}{k}}.
\end{align*}

where the second equality is due to the independence of $R$ and $\Theta$, the third equality holds due to the $\pi$-periodicity of $cos(\theta)$ and $tan(\theta)$, equality $(\star)$ holds since $cos(\pi-x) = - cos(x)$, $tan(\pi-x) = - tan(x)$, equality $(\star \star)$ follows from denoting $\alpha := arctan\prn*{\frac{1}{t}}$, and the last equality holds due to the following reason: the known identity of $tan(x) + tan(\frac{1}{x}) = \frac{\pi}{2}$ implies that\[ \alpha = arctan(2M+1) = \frac{\pi}{2} - \frac{1}{2M+1} + O\prn*{\frac{1}{M^3}},\] and therefore \[sin(2\alpha) = sin\prn*{\pi - \frac{2}{2M+1} - O\prn*{\frac{1}{M^3}}} = sin\prn*{\frac{2}{2M+1} + O\prn*{\frac{1}{M^3}}} = \frac{2}{2M+1} + O\prn*{\frac{1}{M^3}},\] where the last equality is due to the Taylor expansion of $sin(x)$. Together: \[
    \frac{\pi}{2} - \alpha - \frac{sin(2\alpha)}{2} = \frac{\pi}{2} - \frac{\pi}{2} + \frac{1}{2M+1} - O\prn*{\frac{1}{M^3}} - \frac{1}{2M+1} - O\prn*{\frac{1}{M^3}} = O\prn*{\frac{1}{M^3}}.
\]

This completes the proof of \cref{lem:negligible-a_j-minus-b_j}.
\end{proof}

By \cref{lem:negligible-a_j-minus-b_j}, and \cref{eq:bound-norm-of-a-minus-m} we get that
\begin{equation}\label{eq:squared-norm-a-minus-m-lb}
    \E\brk*{\norm{a-m}_2^2} \ge 1 - O\prn*{\frac{1}{k}}.
\end{equation}
The following lemma shows that this implies that $\E\brk*{\norm{a-m}_2} \ge 1 - O\prn*{\frac{1}{k}}$:
\begin{lem}\label{lem:norm-a-minus-m-lb}
    \[
    \E\brk*{\norm{a-m}_2} \ge 1 - O\prn*{\frac{1}{k}}.
    \]
\end{lem}
\begin{proof}
Let $Z := \norm{a - m}_2$.
We already have the second moment lower bound $\E\brk*{Z^2} \ge 1 - O\prn*{\frac{1}{k}}$, and the goal is to show that $\E\brk*{Z} \ge 1 - O\prn*{\frac{1}{k}}$.
By the definition of the variance of $Z$: 
\begin{equation}\label{eq:var-second-moment-squared-expected-val}
    \prn*{\E[Z]}^2 = \E[Z^2] - Var(Z).
\end{equation}
The idea in the rest of the lemma is to show that $Var(Z) = O\prn*{\frac{1}{k}}$ and thus by \cref{eq:squared-norm-a-minus-m-lb}, \cref{eq:var-second-moment-squared-expected-val} we'll get the desired.

\[
    Var(Z) = \E\brk*{(Z - \E\brk*{Z})^2} = \norm{Z - \E[Z]}_{L^2}^2,
\]
where the $L_2$ norm of a random variable $X$ is defined as $\norm{X}_{L^2} := \prn*{\E \brk*{|X|^2}}^{\nf{1}{2}}$ (see Chapter 1.3 in \cite{Vershynin2025HDP}).

Proposition 2.6.6 of \citep{Vershynin2025HDP} implies that for any sub-gaussian random variable $X$: $\norm{X}_{L^2} \le O(1) \cdot \norm{X}_{\psi_2}$, where $\norm{\cdot}_{\psi_2}$ is the sub-gaussian norm defined as \[\norm{X}_{\psi_2} = \inf \crl*{K > 0 \mid \E\brk*{exp\prn*{\frac{X^2}{K^2}}} \le 2}.\]
We thus get
\begin{equation}\label{eq:bound-variance-by-subgauss-norm}
    Var(Z) \le O \prn*{\norm{Z - \E \brk*{Z}}_{\psi_2}^2}.
\end{equation}

And so, we turn to getting a bound on $\norm{Z - \E \brk*{Z}}_{\psi_2}$.

For any $R \in SO(d)$, let $f(R): SO(d) \to \R$ be the function $f(R) = \norm{a(R) - m(R)}_2$. So $f(R) = Z$.
\begin{clm}\label{clm:f-is-M-Lipchitz}
    $f$ is $O(M)$ Lipschitz w.r.t. metric measure space $\prn*{SO(d), \norm{\cdot}_F, \mathbb{P}}$.\footnote{$\norm{\cdot}_F$ is the Frobenius norm; for any matrix $A_{m \times n}$ : $\norm{A}_F := \sqrt{\sum_{i,j \in [m]\times[n]} A_{i,j}^2}$.}
\end{clm}
We defer the proof of \cref{clm:f-is-M-Lipchitz} to \cref{sec:proof-of-clm-f-is-M-Lipschitz}.

Theorem 5.2.7 of 
\citep{Vershynin2025HDP} implies the following concentration inequality of $f$ around its expected value:
\begin{equation}\label{eq:bound-sub-gaussian-norm-of-f-minus-expected-f}
    \norm{f(R) - \E \brk*{f(R)}}_{\psi_2} \le \frac{C L}{\sqrt{d}},
\end{equation}
where $C$ is a constant, $L$ is the Lipschitz constant of $f$ w.r.t. metric measure space $\prn*{SO(d), \norm{\cdot}_F, \mathbb{P}}$. \cref{clm:f-is-M-Lipchitz}, together with \cref{eq:bound-sub-gaussian-norm-of-f-minus-expected-f}, yield \[
    \norm{f(R) - \E \brk*{f(R)}}_{L^2} \le O\prn*{\frac{L}{\sqrt{d}}} \le O \prn*{\frac{1}{\sqrt{k}}},
\]
where the last inequality holds due to plugging in $L = O(M) = O(\sqrt{k})$ (by \cref{clm:f-is-M-Lipchitz}) and $k = \floor{\sqrt{d}}$.

By the definition of $f$ and \cref{eq:bound-variance-by-subgauss-norm}:
$Var(Z) \le O \prn*{\frac{1}{k}}$. Plugging this fact in \cref{eq:var-second-moment-squared-expected-val} completes the proof of this lemma.

\end{proof}

By \cref{lem:norm-a-minus-m-lb} and \cref{eq:alg-expected-val-sqrt2-lb-of-randcwmed}: $\E[ALG] \ge 2k - O(1)$. Since $OPT \le \sqrt{2}k + O(1)$ we get the desired:
\[
    \frac{\E\brk{ALG}}{OPT} \ge \sqrt{2} - O\prn*{\frac{1}{k}} = \sqrt{2} - O\prn*{\frac{1}{\sqrt{d}}}.
\]

\end{proof}

\subsection{Proof of \cref{clm:Dj-equiv-condition}}\label{sec:proof-of-clm-d-j-equiv-cond}
\begin{proof}
$c_j \in [a_j, b_j]$ iff $(a_j - c_j)(b_j - c_j) \le 0$. By writing $s := a_j + b_j$, $d := a_j - b_j$, we get $a = \frac{s+d}{2}, b = \frac{s-d}{2}$. Since $c_j = -M s$:
\[
    a_j - c_j = \frac{s+d}{2} + Ms = \frac{(2M + 1)s + d}{2}, \quad b_j - c_j = \frac{s-d}{2} + Ms = \frac{s(2M+1) - d}{2}.
\]
Thus $c_j \in [a_j, b_j]$ iff 
\[
    \prn*{\frac{(2M + 1)s + d}{2}}\prn*{\frac{s(2M+1) - d}{2}} \le 0,
\]
or equivalently iff $(2M+1)\abs{s} \le \abs{d}$, which implies:

\[
c_j \in [a_j, b_j] \iff \abs{a_j + b_j} \le \frac{\abs{a_j - b_j}}{2M + 1},
\]
\end{proof}

\subsection{Proof of \cref{clm:tech-clm-event-Dj}}\label{sec:proof-clm-tech-clm-event-Dj}
\begin{proof}
We wish to show:
\[
    \prn*{a_j - m_j}^2 + \prn*{b_j - m_j}^2 = (a_j - b_j)^2 + 2\prn*{\prn*{m_j - \frac{a_j + b_j}{2}}^2 - \frac{\prn*{a_j-b_j}^2}{4}} \indic{D_j}.
    \]

For the sake of abbreviation let $x = a_j, z = b_j$ and $y = c_j$, $m = \frac{x+z}{2}$ and $r = \frac{\abs{x-z}}{2}$.
We now have two cases.
In the first case, $y \notin [x,z]$. In this case $m_j \in \crl*{x,z}$ which implies that 
$(x-m_j)^2 + (z-m_j)^2 = (x-z)^2$ (as one of the terms cancels out when plugging in the value of $m_j$).

In the second case is of event $D_j$.
Let $s := \frac{x-z}{2}$. In this case:
\begin{align*}
    (x-m_j)^2 + (z-m_j)^2 &  = (x-y)^2 + (z-y)^2 \\
    & = \prn*{m + s - y}^2 + \prn*{m - s - y}^2 \\
    & = \prn*{\prn*{m-y}^2 + 2s(m-y) + s^2} \\
    & \ \ \ + \prn*{\prn*{m-y}^2 - 2s(m-y) + s^2} \\
    & = \prn*{\prn*{(m-y)^2 + s^2}} \\
    & = 4s^2 + 2\prn*{\prn*{m-y}^2 - s^2} \\
    & = (x-z)^2 + 2\prn*{\prn*{m-y}^2 - r^2} \numberthis \label{eq:tmp1},
\end{align*}
where the last equality is due to the definition of $s$.


\end{proof}

\subsection{Proof of \cref{clm:convergence-to-gaussians}}\label{sec:proof-of-clm-convergence-to-gauss}
\begin{proof}
Let $g(u,v): \R \times \R \to \R$ be $g(u,v) := v^2 \indic{\abs{u} \le t \abs{v}}$. For any $L > 0$ let $g_L := \min \crl*{g,L}$.
Let $(X,Y)$ be a coupling of $(U_1,V_1),(G_1,G_2)$ such that $\P(X \neq Y) = d_{\mathrm{TV}}((U_1,V_1),(G_1,G_2))$ (there is such coupling, see \cite{levin2017markov} Proposition 4.7 and remark 4.8 for more information).
\begin{align*}
    \abs*{\E \brk*{g_L(U_1, V_1)} - \E \brk*{g_L(G_1, G_2)}} &= \abs*{\E \brk*{g_L(U_1, V_1) - g_L(G_1, G_2)}} \\
    & \le \E\brk*{\abs{g_L(U_1, V_1) - g_L(G_1, G_2)}} \le 2 L \P(X \neq Y) \\
    & = 2 L \cdot d_{\mathrm{TV}} \prn*{((U_1,V_1), (G_1,G_2))} = O\prn*{\frac{L}{d}}, \numberthis \label{eq:O-L-over-d-via-TV-bound}
\end{align*}
where the first inequality is due to Jensen's inequality.

Let $h_L: \R \times \R \to \R$ be $h_L := g - g_L$. So
\begin{align*}
    \abs{\E \brk*{g(U_1, V_1)} - \E \brk*{g(G_1, G_2)}} \\
    & = \abs{\E \brk*{h_L(U_1, V_1)} + \E \brk*{h_L(G_1, G_2)} \\
    & \ \ \ + \E \brk*{g_L(U_1, V_1)} - \E \brk*{g_L(G_1, G_2)}} \\
    & \le \frac{3}{L} + \frac{3}{L}  + O\prn*{\frac{L}{d}}, \numberthis \label{eq:bounding-diff-between-Haar-entries-and-iid-normals}
\end{align*}
where the inequality is due to \cref{eq:O-L-over-d-via-TV-bound} and due to the following claim:
\begin{clm}
    $\E \brk*{h_L(U_1, V_1)}, \E \brk*{h_L(G_1, G_2)} \le \frac{3}{L}$
\end{clm}
\begin{proof}

For any nonnegative random variable $X$ and $L > 0$ it must be that \[
    (X-L)_+ \le X \indic{X > L} \le \frac{X^2}{L} \indic{X > L},
\]
where the last step holds as given that $X > L$ it follows that $X < \frac{X^2}{L}$.
Applying this with $X = g(U_1, V_1)$, together with the fact that $g - g_L = (g-L)_{+}$ (as $g_L = \min\crl*{g,L}$) yields:
\[
    (g-L)_+ \le (v^2 - L)_+ \le \frac{v^4}{L} \indic{v^2 > L},
\]
where the first inequality is due to the fact that $g(u,v) \le v^2$.
By taking expectations:
\[
    0 \le \E\brk*{g(U_1,V_1) - g_L(U_1,V_1)} \le \E\brk*{V_1^2 \ \indic{V_1^2 > L}} \le \frac{\E[V_1^4]}{L} \numberthis \label{eq:fourth-moment-tmp-bound},
\]
where the last inequality holds as if $V_1^2 < L$ then $V_1^2 \ \indic{V_1^2 > L} = 0$ and otherwise \[V_1^2 \ \indic{V_1^2 > L} = V_1^2 \le V_1^2 \cdot \frac{V_1^2}{L} = \frac{V_1^4}{L}.\]

It is known (see \cite{meckes2019random}, Chapter 2, equation (2.4)) that for any Haar entry (such as $V_1$) 
\[  
    sup_d \E[V_1^4] \le 3.
\]
Hence, by \cref{eq:fourth-moment-tmp-bound} we get:
\[\E\brk*{g(U_1,V_1) - g_L(U_1,V_1)} \le \frac{3}{L},\] as desired.

Since $\E[G_2^4] = 3$ for $G_2 \sim N(0,1)$, the same argument above implies
that 
\[
    \E \brk*{h_L(G_1, G_2)} \le \frac{3}{L}
\]
also holds.
\end{proof}

By plugging $L = \sqrt{d}$ in \cref{eq:bounding-diff-between-Haar-entries-and-iid-normals} we indeed get that 
$\E \brk*{g(U_1,V_1)} \ge \E\brk*{g(G_1,G_2)} - O\prn*{\frac{1}{\sqrt{d}}}$, which completes the proof of \cref{clm:convergence-to-gaussians}.
\end{proof}

\subsection{Proof of \cref{clm:f-is-M-Lipchitz}}\label{sec:proof-of-clm-f-is-M-Lipschitz}

\begin{proof}
For any $R, R' \in SO(d)$: Let $a = a(R) = Re_1$, $b = b(R) = R e_2$,$c = c(R) = R(e_1 + e_2)$, $a' = a(R') = R' e_1$, $b' = b(R') = R' e_2$, $c' = c(R') = R'(e_1+e_2)$. Also, $m = m(R) = \rcwmed(RP)$ and $m' = m(R') =\rcwmed(R'P)$. By triangle inequality:

\begin{align*}
    \abs{f(R) - f(R')} &= \abs{\norm{a-m}_2 - \norm{a' - m'}_2} \le \norm{(a-m) - (a' - m')}_2 \\
    & \le \norm{a - a'}_2 + \norm{m - m'}_2. \numberthis \label{eq:Lipchitz-bound-tmp-1}
\end{align*}

We bound each term in RHS.

The first term is $\norm{a-a'}_2$:

\begin{equation}\label{eq:bound-on-a-minus-a-tag-norm}
    \norm{a-a'}_2 = \norm{(R-R')e_1}_2 \le \norm{R - R'}_F,
\end{equation}

where $\norm{\cdot}_F$ is the Frobenius norm and thus the inequality holds as for any unit vector $v$ and matrix $A$: $\norm{Ax}_2 \le \norm{A}_F$.
Similarly, 
\begin{equation}\label{eq:bound-on-b-minus-b-tag-norm}
    \norm{b-b'}_2 \le \norm{R - R'}_F.
\end{equation}

We move on to bounding $\norm{m-m'}_2$:

\begin{equation}\label{eq:tmp-bound-on-m-minus-m-tag-sq-norm}
    \norm{m - m'}_2^2 = \sum_{j=1}^d (m_j -m'_j)^2.
\end{equation}
Before proceeding, we give the following claim:
\begin{clm}\label{clm:medians-diff-bund}
    Let $n = 2k+1 \in \N$ be an odd natural number. Let $med(w)$ be the median of $w = (w_1,\ldots,w_n) \in \R^n$. Then for any $x,y \in \R^n$: \[
    \abs{med(x) - med(y)} \le \norm{x - y}_{\infty}.
    \]
\end{clm}
\begin{proof}
Let $\delta = \norm{x-y}_{\infty}$. So $y_i \in [x_i-\delta,x_i + \delta]$ for any $i \in [n]$. 
Let $I^- = \crl*{i \mid x_i \le \med(x)}$. So $\abs{I^-} \ge \frac{n+1}{2}$. For any $i \in I^{-}$: 
$y_i \le x_i + \delta \le med(x) + \delta$.
Hence at least half the $y_i$ values are upper bounded by $med(x) + \delta$, and thus $med(y) \le med(x) + \delta$. Similarly for $I^+ = \crl*{i \mid x_i \ge med(x)}$, $\abs{I^+} \ge \frac{n+1}{2}$ and for all $i \in I^+$: $y_i \ge med(x) - \delta$ and thus $med(y) \ge med(x) - \delta$. Together we get 
\[
    \abs{med(y) - med(x)} \le \delta,
\] as desired.

\end{proof}

From \cref{clm:medians-diff-bund}, for any $j \in [d]$: $\abs{m_j -m'_j} \le \max\crl*{|a_j - a'_j|, |b_j - b'_j|, |c_j - c'_j|}$. Together with \cref{eq:tmp-bound-on-m-minus-m-tag-sq-norm} we get:
\begin{align*}
    \norm{m - m'}_2^2 & \le \sum_{j=1}^d \prn*{\max\crl*{|a_j - a'_j|, |b_j - b'_j|, |c_j - c'_j|}}^2 \le \sum_{j=1}^d (a_j - a'_j)^2 + (b_j - b'_j)^2 + (c_j - c'_j)^2 \\
    & = \norm{a - a'}_2^2 + \norm{b - b'}_2^2 + \norm{c - c'}_2^2 \\
    & \le \norm{R - R'}_F^2 + \norm{R - R'}_F^2 + (M \sqrt{2} \norm{R - R'}_F)^2 = (2 + 2M^2)\norm{R - R'}_F^2, \numberthis \label{bound-on-norm-m-minus-m-tag-squared}
\end{align*}
where the second inequality is due to the fact that for any $\alpha, \beta, \gamma \in \R$: $\prn*{\max{\alpha, \beta, \gamma}}^2 \le \alpha^2 + \beta^2 + \gamma^2$, the third inequality holds due to \cref{eq:bound-on-a-minus-a-tag-norm}, \cref{eq:bound-on-b-minus-b-tag-norm} and the following claim that gives the bound on $\norm{c-c'}_2$:
\begin{clm}
    $\norm{c-c'}_2 \le M \sqrt{2} \norm{R - R'}_F$.
\end{clm}
\begin{proof}
Since $c = -M(a+b)$: 
$c - c' = -M((a-a') + (b-b'))$ and thus
\begin{align*}
    \norm{c - c'}_2 & = M \norm{(a - a') + (b-b')}_2 \le M \prn*{\norm{a-a'}_2 + \norm{b - b'}_2} \\
    & \le M \sqrt{2} \sqrt{\norm{a-a'}_2^2 + \norm{b - b'}_2^2} \\
    & \le M \sqrt{2} \norm{R - R'}_F, \numberthis \label{eq:bound-for-c-minus-c-tag}
\end{align*}

where the first inequality is due to triangle inequality, the second is due to the standard inequality $x + y \le \sqrt{2}\sqrt{x^2 + y^2}$ that applies for any $x,y \ge 0$, and the last inequality holds since 
\begin{align*}
    \norm{a-a'}_2^2 + \norm{b-b'}_2^2 & = \norm{(R-R')e_1}_2^2 + \norm{(R-R')e_2}_2^2  \\
    & \le \sum_{i=1}^d \norm{(R-R') e_i}_2^2= \sum_{i=1}^d \sum_{j=1}^d (R-R')_{j,i}^2  = \norm{R - R'}_F^2.
\end{align*}
\end{proof}

Plugging \cref{eq:bound-on-a-minus-a-tag-norm} and \cref{bound-on-norm-m-minus-m-tag-squared} into \cref{eq:Lipchitz-bound-tmp-1} gives:
\[
    \abs{f(R) - f(R')} \le \norm{R - R'}_F + \sqrt{2 + 2M^2}\norm{R - R'}_F = O(M) \cdot \norm{R - R'}_F,
\]
which concludes the proof of the claim.
\end{proof}
\section{Missing analysis from \cref{sec:predictions-learning-aug}}\label{sec:appendix-learning-aug}

\subsection{Proof of \cref{lem:geometric-median-approx-robustness}}

First, we give a lemma that bounds the ``distance robustness'' of the $\OneMed$: that is, how far may the $\OneMed$ get from its original location as a result of an ``insertion corruption''.
\begin{lem}\label{lem:geomed-distance-robustness-given-added-points}
    Let $c \in (0,1)$ and $(V,d)$ be a metric space. Let $P, \hat{P} \subset V$ be two multisets of points such that $\abs{\hat{P}} = c \abs{P}$. Then: 
    \[
    d\prn*{med\prn*{P},med \prn*{P \cup \hat{P}}} \le \frac{2}{|P|\prn*{1-c}}  C\prn*{P, med \prn*{P}},
    \]
    where , and $med(X)$ is the $\OneMed$ of $X$: $med(X) := \argmin_{t \in V} C(X,t)$.
\end{lem}

\begin{proof}
Let $n := \abs{P}$ and let $P' = P \cup \hat{P}$ which implies $\abs{P'} = (1+c)n$. Let $m = \argmin_{t \in V} C(P,t)$ and $m' = \argmin_{t \in V} C((P \cup \hat{P}),t)$. 

By triangle inequality:

\begin{align*}
    & \sum_{p \in P} \prn*{d(m, m') - d(m, p)} + \sum_{\hat{p} \in \hat{P}} \prn*{d(m,\hat{p}) - d(m, m')} \\
    & \qquad \le \sum_{p \in P} d(p, m') + \sum_{\hat{p} \in \hat{P}} d(\hat{p}, m') = \sum_{p' \in P'} d(p', m') \\
    & \qquad \le \sum_{p' \in P'} d(p', m) = \sum_{p \in P} d(p, m) + \sum_{\hat{p} \in \hat{P}} d(\hat{p}, m),
\end{align*}
where the last inequality holds as the definition of $m'$ implies: $\sum_{p' \in P'} d(p', m') \le \sum_{p' \in P'} d(p', m)$.
It follows from the above that $\prn*{\abs{P} - \abs{\hat{P}}}d(m,m') \le 2 \sum_{p \in P} d(p, m)$, or equivalently:
\begin{equation*}
    d(m,m') \le \frac{2}{n(1-c)}\sum_{p \in P} d(p, m).
\end{equation*}
\end{proof}

We are now ready to give the proof of \cref{lem:geometric-median-approx-robustness}:
\begin{proof}
    Let $n := \abs{P}$ and let $P' = P \cup \hat{P}$ which implies $\abs{P'} = (1+c)n$. Let $m = \argmin_{t \in V} C(P,t)$ and $m' = \argmin_{t \in V} C((P \cup \hat{P}),t)$. 
    \begin{align*}
        \sum_{p \in P} d(p, m') & = \sum_{p \in P \cup \hat{P}} d(p, m') - \sum_{p \in \hat{P}} d(p, m') \\
        & \le \sum_{p \in P \cup \hat{P}} d(p, m) - \sum_{p \in \hat{P}} d(p, m') \\
        & = C \prn*{P, m} + \sum_{p \in \hat{P}} d(p,m) - d(p,m') \\
        & \le C \prn*{P, m} + \sum_{p \in \hat{P}} d(m,m') \\
        & \le \prn*{1 + \frac{2c}{1-c}} C\prn*{P, m},
    \end{align*}
    where the first inequality holds by the definition of $m'$, the second inequality is due to triangle inequality, and the last inequality follows from \cref{lem:geomed-distance-robustness-given-added-points}.
\end{proof}

\section{Missing proofs from \cref{sec:lower-bounds}}\label{sec:missing-proofs-rand-dict-lb}

\subsection{Proof of \cref{clm:trigo-identity}}
\begin{proof}
Let $S \;=\;\sum_{k=1}^{n-1}\sin\prn*{kx}, \ x = \frac{\pi}{n}$. We will show $S = \cot\prn*{\tfrac{x}{2}}$.

Multiply both sides by $\sin\prn*{\tfrac{x}{2}}$:
\[
S\,\sin \prn*{\tfrac{x}{2}}
=\sum_{k=1}^{n-1}\sin\prn*{kx}\,\sin \prn*{\tfrac{x}{2}}.
\]

Apply the product-to-sum identity
\[
2\,\sin\alpha\,\sin\beta
=\cos\prn*{\alpha - \beta}\;-\;\cos\prn*{\alpha + \beta},
\]
with $\alpha = kx$ and $\beta = \tfrac{x}{2}$.  Then
\[
\sin\prn*{kx}\,\sin \prn*{\tfrac{x}{2}}
=\tfrac12 \brk*{\cos\prn*{(k-\tfrac12)x}
            -\cos\prn*{(k+\tfrac12)x}}.
\]
Hence
\[
S\,\sin \prn*{\tfrac{x}{2}}
=\frac12
\sum_{k=1}^{n-1}
\Bigl[\cos\prn*{(k-\tfrac12)x}-\cos\prn*{(k+\tfrac12)x}\Bigr].
\]

Observe that the sum on the right telescopes:
\[
\begin{aligned}
&\brk*{\cos\prn*{\tfrac12x}-\cos\prn*{\tfrac32x}}
+\cdots
+\brk*{\cos\prn*{(n-\tfrac32)x}-\cos\prn*{(n-\tfrac12)x}}\\
&\quad=\cos \prn*{\tfrac12x}-\cos \prn*{(n-\tfrac12)x}.
\end{aligned}
\]
Therefore
\[
S\,\sin \prn*{\tfrac{x}{2}}
=\tfrac12\brk*{\cos \prn*{\tfrac{x}{2}}
              -\cos \prn*{(n-\tfrac12)x}}.
\]

Substitute $x=\tfrac{\pi}{n}$. We get $(n-\tfrac12)x
=\pi-\tfrac{\pi}{2n}$ which implies \[\cos \prn*{(n-\tfrac12)x}
=-\cos \prn*{\tfrac{\pi}{2n}}.\]
Thus
\[
S \sin \prn*{\tfrac{\pi}{2n}}
=\tfrac12\bigl[\cos\prn*{\tfrac{\pi}{2n}} - (-\cos\prn*{\tfrac{\pi}{2n}})\bigr]
=\cos \prn*{\tfrac{\pi}{2n}}.
\]

Divide both sides by $\sin \prn*{\tfrac{\pi}{2n}}$:
\[
S
=\frac{\cos \prn*{\tfrac{\pi}{2n}}}
     {\sin \prn*{\tfrac{\pi}{2n}}}
=\cot \prn*{\tfrac{\pi}{2n}}.
\]
This completes the proof.
\end{proof}

\subsection{Proof of \cref{lem:random-sphere-vector-distance-from-e1}}\label{sec:proof-of-expected-rand-sphere-vectors-lemma}
\begin{proof}
\begin{align*}
    \norm{X-e_1}_2 &= \sqrt{\langle X-e_1,X-e_1 \rangle} = \sqrt{\norm{X}_2^2 + 1 - 2  \langle X,e_1 \rangle} \\
    & = \sqrt{2-2X_1} = \sqrt{2}\ \sqrt{1-X_1}. \numberthis \label{tmp:X-minus-e_1-norm}
\end{align*}

Recall the generalized binomial theorem: for $\alpha \in \mathbb{R}$ and $|t| \le 1$,
\[
(1 - t)^{\alpha} = \sum_{n=0}^\infty \binom{\alpha}{n} (-t)^n.
\]
For $\alpha = \frac{1}{2}$, we thus have:
\[
(1 - t)^{1/2} = 1 - \frac12 t - \frac18 t^2 - \frac{1}{16} t^3 - \frac{5}{128} t^4 + O\prn*{t^5}. \numberthis \label{eq:1-minus-t-power-of-half}
\]

Taking expectations in \cref{tmp:X-minus-e_1-norm}, plugging in \cref{eq:1-minus-t-power-of-half} and noting that odd moments vanish (proved below), we get
\[
\mathbb{E} \brk*{\norm{X-e_1}_2}
=\sqrt{2} \prn*{1-\frac18\,\mathbb{E}[X_1^2]
-\frac{5}{128}\,\mathbb{E}[X_1^4]
-\frac{21}{1024}\,\mathbb{E}[X_1^6]-\cdots}.
\]

To see why the odd moments vanish, consider the orthogonal transformation $R=\mathrm{diag}(-1,1,\dots,1)$. By invariance of $dS$ to orthogonal isometries:
\[ \int_{S^{d-1}} x_1^{2k-1}\,dS
=\int_{S^{d-1}} \prn*{Rx}_1^{2k-1}\,dS
=\int_{S^{d-1}} \prn*{-x_1}^{2k-1}\,dS
=-\int_{S^{d-1}} x_1^{2k-1}\,dS,
\]
hence $\E\brk*{X_1^{2k-1}} = \int_{S^{d-1}} x_1^{2k-1}\,dS = 0$.
\end{proof}
\section{$\rrcwmed$ and $\projmed$ connection}\label{sec:proj-median-connection}

In this section we explain the connection between the $\rrcwmed$ expected approximation ratio and the approximation of a related median generalization --- the $\projmed$.

The $\projmed$ is the average median of the projection of the dataset onto all possible lines that go through the origin. Formally:
\begin{definition}[Projection Median]
Let $P \subset \R^d$ be a finite non-empty set of points.  
For a unit vector $u \in \mathbb{S}^{d-1}$, let $P_u$ denote the projection of $P$ onto the line through the origin in direction $u$, 
and let $\operatorname{med}(P_u)$ be the one-dimensional median of $P_u$. The $\projmed$ of $P$ is defined as
\[
\projmed(P) = d \int_{\mathbb{S}^{d-1}} median(P_u)\, d\mu(u),
\]
where $\mu$ is the normalized uniform probability measure on the unit sphere $\mathbb{S}^{d-1}$.

\end{definition}

\cite{durocher2009projection,basu2012projection} also give an equivalent definition:
\begin{prop}[Projection Median and $\rcwmed$ equivalence]\label{prop:proj-med-equiv-def}

    Let $P \subset \mathbb{R}^d$ be a finite non-empty set of points.
    The $\projmed$ of $P$ is equivalently given by
    \[
        \projmed(P) \;=\; \int_{SO(d)} \rcwmed(P,A) \, d\nu(A),
    \]
    where $\nu$ is the normalized Haar measure on the rotation group $SO(d)$.
    
\end{prop}

\cref{prop:proj-med-equiv-def} implies that the $\projmed$ is equivalent to the average case (over all rotation matrices) of $\rcwmed$. Note, however, that it does not mean that the two mechanisms have the same approximation ratio.
Indeed, consider the following simple example:
\begin{example}\label{ex:simple-example-diff-projmed-rrcwmed}
    Let $P = \crl*{(1,0), (0,1), (-1,0)}$.
\end{example}
\begin{clm}\label{clm:ex-simple-example-diff-projmed-rrcwmed}
    For \cref{ex:simple-example-diff-projmed-rrcwmed}, the approximation ratio of $\projmed$ is strictly smaller than the expected approximation ratio of $\rrcwmed$ by a multiplicative factor strictly bigger than $1$.
\end{clm}
We defer the proof of \cref{clm:ex-simple-example-diff-projmed-rrcwmed} to \cref{sec:proof-of-clm-ex-simple-example-diff-projmed-rrcwmed}.

However, we do get the following connection:
\begin{prop}\label{prop:proj-median-ar-le-rrcwmed-ar}
    Let $AR_{proj}$, $AR_{rr}$ be the approximation ratios w.r.t. the social cost ($\scost$), of
    
    $\projmed$ and $\rrcwmed$ respectively. Then $AR_{proj} \le AR_{rr}$.
\end{prop}
\begin{proof}
    Let $P \subseteq \R^d$.
    \begin{align*}
        \scost&\prn*{P,\projmed(P)} \\
        & \quad \quad \ \ \ = \scost\prn*{P,\int_{SO(d)} \rcwmed(P,A) \, d\nu(A)} \\
        & \quad \quad \ \ \ \le \int_{SO(d)} \scost\prn*{P, \rcwmed(P,A)\, d\nu(A)} \\
        & \quad \quad \ \ \ = \E[\scost\prn*{P,\rrcwmed(P)}],
    \end{align*}
    where the second equality holds due to \cref{prop:proj-med-equiv-def} and the inequality is due to Jensen's inequality (which applies since $\scost$ is convex).
    The desired follows from the fact that \[AR_{proj} = \frac{\scost\prn*{P,\projmed(P)}}{OPT}, \qquad AR_{rr} = \frac{\E[\scost\prn*{P,\rrcwmed(P)}]}{OPT}.\]
\end{proof}

\cref{prop:proj-median-ar-le-rrcwmed-ar} implies that (1) any lower bound on the approximation ratio of $\projmed$ also holds for $\rrcwmed$ and any upper bound on $\rrcwmed$ also holds for $\projmed$.

Unfortunately, computing the $\projmed$ is not strategyproof, as mentioned in the following remark:
\begin{remark}\label{rem:projmed-not-sp}
    The $\projmed$ is deterministic, anonymous and has an approximation ratio strictly smaller than $\sqrt{2}$ in $\R^2$. Thus, due to the lower bound of \cite{goel2023optimality}, it is not strategyproof (see \cref{prop:cwmed-properties}).
\end{remark}

This connection, however, does imply a significantly narrowing of the approximation gap left by \citet{basu2012projection}. \citet{basu2012projection} show that the approximation ratio is upper bounded by $ \frac{d}{\pi} B(\frac{d}{2}, \frac{1}{2}) = \Omega(\sqrt{d})$, where $B(\alpha,\beta)$ denotes the beta function. 
\cref{prop:proj-median-ar-le-rrcwmed-ar}, together with \cref{thm:rrcwmed-expected-approximation-ratio-Rd} directly imply that the approximation ratio of $\projmed$ is upper bounded by a constant ($\approx 1.547$).

Completely closing the approximation gap of the projected median stays an open 
problem.
For $\R^2$, the gap is known to be in $\brk*{\sqrt{\frac{4}{\pi^2} + 1},\frac{4}{\pi}}$ \cite{durocher2009projection}. 

\cref{lem:rand-cwmed-lb} (the lower bound of our analysis for the $\rrcwmed$) implies that no further improvements may be made via analyzing $\rrcwmed$, and a different technique must be found. 

For $\R^d$, 
\cref{thm:rrcwmed-expected-approximation-ratio-Rd} also implies that no upper bound better than $\sqrt{2}$ can be obtained via analyzing the expected approximation ratio of $\rrcwmed$.

Strictly separating between the expected approximation ratio of $\rrcwmed$ and the approximation ratio of $\projmed$ remains an open problem left for future work.

\subsection{Proof of \cref{clm:ex-simple-example-diff-projmed-rrcwmed}}\label{sec:proof-of-clm-ex-simple-example-diff-projmed-rrcwmed}

\begin{proof}

We will show that the projected median is $m_{proj} = (\frac{1}{2}, 0)$ and its expected cost is $\sqrt{5} + \frac{1}{2}$. Then, we show that the expected approximation ratio of the $\rrcwmed$ mechanism is bigger, and thus we show that indeed the two can be different.

Let $R_\theta$ be the rotation matrix by $\theta$, and let $p'_i$ be point $p_i$ rotated by $R_\theta$ where $p_1 = (1,0)$, $p_2 = (0,1)$ and $p_3 = (-1,0)$. Then
\begin{align*}
R_\theta &= \begin{pmatrix}\cos\theta & -\sin\theta\\ \sin\theta & \cos\theta\end{pmatrix},\qquad
\theta \in [0,\tfrac{\pi}{2}],\\
p_1' &= (\cos\theta,\ \sin\theta),\quad
p_2' = (-\sin\theta,\ \cos\theta),\quad
p_3' = (-\cos\theta,\ -\sin\theta).
\end{align*}
If $\theta \in [0,\frac{\pi}{4})$ the $x$ coordinates are $\crl*{-cos(\theta), -sin(\theta), cos(\theta)}$ and the $y$ coordinates are

$\crl*{-sin(\theta), sin(\theta), cos(\theta)}$. Thus, the resulting $\cwmed$ of the rotated points would be

$m'_{\theta} = (-sin(\theta), sin(\theta))$.
If $\theta \in [\frac{\pi}{4},\frac{\pi}{2})$ the $x$ coordinates are $\crl*{-sin(\theta), -cos(\theta), cos(\theta)}$ and the $y$ coordinates are $\crl*{-sin(\theta), cos(\theta), sin(\theta)}$. Thus, the resulting $\cwmed$ of the rotated points would be $m'_{\theta} = (-cos(\theta), cos(\theta))$.

This implies that for $m(\theta) := \min\{\sin\theta,\ \cos\theta\}$: $m'_\theta = \big(-\,m(\theta),\ m(\theta)\big)$, which means that 
$\rcwmed(P, \theta) = R_{-\theta}\,m'_\theta
= m(\theta)\,\big(\sin\theta-\cos\theta,\ \sin\theta+\cos\theta\big)$.

We can now compute the cost of $\projmed$, denoted by $m_{proj}$. By the proof of Theorem 4 of \cite{durocher2009projection}, the $\projmed$ equals:
\begin{align*}
m_{proj}(P)
&= \frac{2}{\pi}\int_{0}^{\pi/2} \rcwmed(P, \theta)\,d\theta \\
&= \frac{2}{\pi}\Bigg[
\int_{0}^{\pi/4} \sin\theta\,(\sin\theta-\cos\theta,\ \sin\theta+\cos\theta)\,d\theta
\;\\
& \ \ \ \ \ \ \ \ +\;
\int_{\pi/4}^{\pi/2} \cos\theta\,(\sin\theta-\cos\theta,\ \sin\theta+\cos\theta)\,d\theta
\Bigg].
\end{align*}
Therefore, we may compute $m_{proj} = (x(m_{proj}), y(m_{proj}))$:
\begin{align*}
x(m_{proj})
&= \frac{2}{\pi}\Bigg[
\int_{0}^{\pi/4}\!\big(\sin^2\theta-\sin\theta\cos\theta\big)\,d\theta
\;+\;
\int_{\pi/4}^{\pi/2}\!\big(\sin\theta\cos\theta-\cos^2\theta\big)\,d\theta
\Bigg]
\\
& = 0,\\[4pt]
y(m_{proj})
&= \frac{2}{\pi}\Bigg[
\int_{0}^{\pi/4}\!\big(\sin^2\theta+\sin\theta\cos\theta\big)\,d\theta
\;+\;
\int_{\pi/4}^{\pi/2}\!\big(\sin\theta\cos\theta+\cos^2\theta\big)\,d\theta
\Bigg]
\\
& = \tfrac{1}{2}.
\end{align*}
Hence $m_{proj}(P) = (0,\tfrac{1}{2})$ and thus:
\begin{align*}
\scost(m_{proj}) &= \|(1,0)-(0,\tfrac12)\| + \|(0,1)-(0,\tfrac12)\| + \|(-1,0)-(0,\tfrac12)\| \\
&= \sqrt{1+\tfrac14} + \tfrac12 + \sqrt{1+\tfrac14}
= \sqrt{5} + \tfrac12 \approx 2.736.
\end{align*}

We now compute the expected cost of the $\rrcwmed$, which is the expected value of over all $\theta$ choices of $\scost(P,\rcwmed(P,\theta))$ . For any $\theta$:
\begin{align*}
\scost(P, \rcwmed(P, \theta))
&= \|p_1' - m'_\theta\| + \|p_2' - m'_\theta\| + \|p_3' - m'_\theta\|.
\end{align*}
For \(0 \le \theta \le \pi/4\) (so \(m(\theta)=\sin\theta\)):
\begin{align*}
\|p_1' - m'_\theta\| &= \cos\theta+\sin\theta = \sqrt{\,1+\sin(2\theta)\,},\\
\|p_2' - m'_\theta\| &= |\cos\theta-\sin\theta| = \sqrt{\,1-\sin(2\theta)\,},\\
\|p_3' - m'_\theta\| &= \sqrt{(\sin\theta-\cos\theta)^2 + (2\sin\theta)^2}
= \sqrt{\,3 - 2\cos(2\theta) - \sin(2\theta)\,}.
\end{align*}
For \(\pi/4 \le \theta \le \pi/2\) (so \(m(\theta)=\cos\theta\)):
\begin{align*}
\|p_1' - m'_\theta\| &= \sqrt{\,3 + 2\cos(2\theta) - \sin(2\theta)\,},\\
\|p_2' - m'_\theta\| &= \sqrt{\,1-\sin(2\theta)\,},\\
\|p_3' - m'_\theta\| &= \sqrt{\,1+\sin(2\theta)\,}.
\end{align*}
Therefore (we have uniform \(\theta\), symmetry gives the factor \(2/\pi\))
\begin{align*}
\E_\theta[& \scost(P,\rcwmed(P,\theta))]
= \\
&  \frac{2}{\pi}\Bigg[
\int_{0}^{\pi/4}
\Big(
\sqrt{\,3 - 2\cos(2\theta) - \sin(2\theta)\,}
+ \sqrt{\,1+\sin(2\theta)\,}
+ \sqrt{\,1-\sin(2\theta)\,}
\Big)\,d\theta \\
& 
+ \int_{\pi/4}^{\pi/2}
\Big(
\sqrt{\,3 + 2\cos(2\theta) - \sin(2\theta)\,}
+ \sqrt{\,1+\sin(2\theta)\,}
+ \sqrt{\,1-\sin(2\theta)\,}
\Big)\,d\theta
\Bigg].
\end{align*}
Approximating the above term yields:
\begin{align*}
    \E[\rrcwmed(P)] & = \E_\theta[\scost(P,\rcwmed(P,\theta))] \\
    & \approx 2.83 > 2.736 = \scost(m_{proj}).
\end{align*}

\end{proof}

\printbibliography

\end{document}